\documentclass[sigconf]{acmart}
\usepackage{booktabs} 
\acmConference[MobiHoc]{ACM MobiHoc conference}{July 2017}{Chennai, India} 
\usepackage{bm}
\usepackage{bbm} 
\usepackage{framed}
\usepackage{color}
\usepackage{algorithm}
\usepackage{algorithmic}
\usepackage{overpic}
\usepackage{framed}
\usepackage{dsfont}
\usepackage{url}
\newcommand\EatDot[1]{}

\newcommand{\editing}[1]{{ #1}}

\begin{document}

\title{Throughput-Optimal Broadcast in Wireless Networks with Point-to-Multipoint Transmissions}

\author{Abhishek Sinha}
\affiliation{
\department {Laboratory for Information and Decision Systems}
\institution{MIT}}
\email{sinhaa@mit.edu}
\author{Eytan Modiano}
\affiliation{
\department {Laboratory for Information and Decision Systems}
\institution{MIT}}
\email{modiano@mit.edu}

% \CopyrightYear{2016}
% \setcopyright{acmcopyright}
% \conferenceinfo{MobiHoc'16,}{July 04-08, 2016, Paderborn, Germany}
% \isbn{978-1-4503-4184-4/16/07}\acmPrice{\$15.00}
% \doi{http://dx.doi.org/10.1145/2942358.2942390}
%\numberofauthors{2}
%% Three authors sharing the same affiliation.
%    \author{
%      \alignauthor Abhishek Sinha\\ 
%      \affaddr{Laboratory for Information and Decision Systems}\\
%      \affaddr{MIT}     
%      \email{sinhaa@mit.edu}
%%      
%%   \and \and     \alignauthor Georgios Paschos\\
%%       \affaddr{Mathematical and Algorithmic Sciences Lab }
%%       \affaddr{Huawei Technologies Co. Ltd} \\    
%%    \email{ georgios.paschos@huawei.com}
%%
%  \and 
%     \alignauthor Eytan Modiano\\  
%     \affaddr{Laboratory for Information and Decision Systems}\\
%      \affaddr{MIT}     
%      \email{modiano@mit.edu}  
%    %  \email{modiano@mit.edu}
%%
%%       \sharedaffiliation
%%       \affaddr{Department of Electrical Engineering and Computer Science }  \\
%%       \affaddr{University of California, Berkeley }   \\
%%       \affaddr{Berkeley, CA 94720-1776 }
%          }
%\maketitle

% newly added 
%\onecolumn
%\linespread{2}

% For ease of drafting, rendering to single column version
%\onecolumn

% \author{
% \IEEEauthorblockN{Abhishek Sinha}
% \IEEEauthorblockA{Laboratory for Information and Decision Systems, Massachusetts Institute of Technology, Cambridge, MA 02139\\
% Email: sinhaa@mit.edu}
% }

%%%%%%%%%%%%%%%%%%%%%%%%%%%%%%%%%%%%%%%%%%%%%%%%%%%%%%%%%%%%%%%%%%%%%%%%%%%%%%%%%%%%%%%%%

\begin{abstract}
 We consider the problem of efficient packet dissemination in wireless networks with point-to-multi-point wireless broadcast channels. We propose a dynamic policy, which achieves the broadcast capacity of the network. This policy is obtained by first transforming the original multi-hop network into a precedence-relaxed virtual single-hop network and then finding an optimal broadcast policy for the relaxed network. The resulting policy is shown to be throughput-optimal for the original wireless network using a sample-path argument. We also prove the \textbf{NP}-completeness of the finite-horizon broadcast problem, which is in contrast with the polynomial time solvability of the problem with point-to-point channels. Illustrative simulation results demonstrate the efficacy of the proposed broadcast policy in achieving the full broadcast capacity with low delay. 
\end{abstract}

\begin{CCSXML}
<ccs2012>
<concept>
<concept_id>10003033.10003068.10003073.10003075</concept_id>
<concept_desc>Networks~Network control algorithms</concept_desc>
<concept_significance>500</concept_significance>
</concept>
<concept>
<concept_id>10003752.10003809.10010047.10010048.10003808</concept_id>
<concept_desc>Theory of computation~Scheduling algorithms</concept_desc>
<concept_significance>300</concept_significance>
</concept>
</ccs2012>
\end{CCSXML}

\ccsdesc[500]{Networks~Network control algorithms}
\ccsdesc[300]{Theory of computation~Scheduling algorithms}
%\printccsdesc

\keywords {Broadcasting, Scheduling, Queueing Theory, Throughput Optimality}

\maketitle
\section{Introduction and Related Work} \label{introduction}
%\begin{introduction} \label{introduction}
%Discuss how this problem is different from Broadcasting in wireless networks with point-to-point channels. This problem is hard while the previous problem was easy. Also, the hardness does not stem from wireless activations (cite Devavrat Shah's hardness paper \cite{shah2011hardness}). It stems from the combinatorial aspect of distributing packets in the network. Also, we need to make explicit distinction of our hardness result with others (e.g. \cite{kuhn2009brief}, it proves hardness for unreliable network). 
%\end{introduction}
The problem of disseminating packets efficiently from a set of source nodes to all nodes in a network is known as the \emph{Broadcast Problem}. Broadcasting is a fundamental network functionality, which is used frequently in numerous practical applications, including military communication \cite{Na:2015:RMS:2836587.2836593}, information dissemination and disaster management \cite{Heinzelman:1999:API:313451.313529}, in-network function computation \cite{dietzfelbinger2004gossiping} and efficient dissemination of control information in vehicular networks \cite{chen2010broadcasting}.

Due to its fundamental nature, the Broadcasting problem in wireless networks has been studied extensively in the literature. As a result, a number of different algorithms have been proposed for optimizing different efficiency metrics. Examples include minimum energy broadcast \cite{maric2004cooperative}, minimum latency broadcast \cite{gandhi2008minimizing}, broadcasting with minimum number of retransmissions \cite{wong2013efficient}, and throughput-optimal broadcast  \cite{massoulie2007randomized}. A comprehensive study of different broadcasting algorithms proposed for Mobile Adhoc networks is presented in \cite{williams2002comparison}. 

A fundamental feature of the wireless medium is the inherent point-to-mutipoint nature of wireless links, where a packet transmitted by a node can be heard by all its neighbors. This feature, also known as the wireless broadcast advantage \cite{cui2007distributed}, is especially useful in network-wide broadcast applications, where the objective is to efficiently disseminate the packets among all nodes in the network. Additionally, because of inter-node interference, the set of simultaneous transmissions in a wireless network is restricted to the set of non-interfering feasible schedules.  Designing a broadcast algorithm which efficiently utilizes the broadcast advantage, while respecting the interference constraints is a challenging problem.

 The problem of throughput optimal multicasting in wired networks has been considered in \cite{swati}. In our recent works \cite{sinha2015throughput} \cite{sinha2016throughput} \cite{sinha2016throughput2}, we studied the problem of throughput optimal broadcasting in wireless networks with directed \emph{point-to-point-links} and designed several efficient broadcasting algorithms. The problem of designing throughput optimal broadcast policy in wireless networks with \emph{point-to-multi-point links} was considered in \cite{towsley2008rate}, where the authors studied a highly restrictive ``scheduling-free'' model, where it is assumed that scheduling decisions are made by a central controller, acting independently of their algorithm. With this assumption, they obtained a randomized packet forwarding scheme, which requires a continuous exchange of control information among the neighboring nodes. This algorithm was proved to be throughput optimal \emph{with respect to the given schedules}, using fluid limit techniques. In this paper, we consider the joint problem of throughput optimal scheduling \emph{and} packet dissemination in wireless networks with point-to-multi-point links. Our approach uses the concept of \emph{virtual network}, that we recently introduced in \cite{sinha2016optimal} for solving the generalized network flow problem with point-to-point links. To the best of our knowledge, this is the first known throughput optimal broadcast algorithm in wireless networks with broadcast advantage. \\
 
 \editing{The main contributions of this paper are as follows:
 \begin{itemize}
  \item We propose an online dynamic policy for throughput-optimal broadcasting in wireless networks with point-to-multipoint links.
  \item We prove the \textbf{NP}-completeness of the corresponding finite horizon wireless broadcast problem. 
  \item We introduce a new control policy and proof technique by combining the stochastic Lyapunov drift theory with the deterministic adversarial queueing theory. This essentially enables us to derive a stabilizing control policy for a multi-hop network by solving the problem on  a simpler \emph{precedence-relaxed} virtual single-hop network. 
  %\item We establish a stronger almost sure $\mathcal{O}(\log t)$ growth rate of the queues as compared to our earlier work \cite{sinha2016optimal}, where we established a weaker almost sure $o(t)$ growth rate of the queues for flow problems without broadcast advantage. 
 \end{itemize}

 }
 \editing{The rest of the paper is organized as follows. In section \ref{model} we describe the system model and formulate the problem. In section \ref{hardness} we prove the hardness of the finite-horizon version of the problem. Next, in section \ref{bcast_algorithms} we derive an optimal control policy for a related relaxed version of the wireless network. This control policy is then applied to the original unrelaxed network in section \ref{physical_queue_stability}, where we show that the resulting policy is throughput-optimal, when used in conjunction with a priority-based packet scheduling policy. In section \ref{simulation}, we demonstrate the efficacy of the proposed policy via numerical simulations. Finally, we conclude the paper in section \ref{conclusion}.  
 
 }

% The central algorithmic idea in this paper, namely the \emph{virtual network} and the \emph{virtual queues} was first introduced in our recent work \cite{sinha2016optimal}, where we studied the problem of throughput-optimal routing and scheduling in networks with point-to-point links. Contrary to \cite{sinha2016optimal}, where the virtual queues were associated with each link, in this paper, we associate the virtual queues with each node, due to the point-to-multi-point nature of the wireless links. The corresponding routing, link scheduling, and the packet scheduling policies are also different in these two papers, although both of them are derived from a common conceptual framework. Moreover, using an improved analysis, here we obtain an almost sure $\mathcal{O}(\log(t))$ bound for the virtual queue size, whereas in \cite{sinha2016optimal}, only a $o(t)$ bound for the virtual queue size was established. Thus, these series of papers firmly establish the efficacy of the virtual network technique and the universality of the proposed \textsf{UMW} policy for a wide range of routing and scheduling problems in wireless networks.  

\section{System Model and Problem Formulation} \label{model}
We consider the problem of efficiently disseminating packets, arriving randomly at source nodes, to all nodes in a wireless network. The system model and the precise problem statement are described below.
\subsection{Network Model}
Consider a wireless network with its topology given by the directed graph $\mathcal{G}(V,E)$. The set $V$ denotes the set of all  nodes, with $|V|=n$. If node $j$ is within the transmission range of node $i$, there is a directed edge $(i,j)\in E$ connecting them. Due to the inherent point-to-multi-point broadcast nature of the radio channel, a transmitted packet can be heard by all out-neighbors of the transmitting node. In other words, the packets are transmitted over the \emph{hyperedges}, where a hyperedge is defined to be the union of all outgoing edges from a node. The system evolves in a slotted time structure. External packets, which are to be broadcasted throughout the network, arrive at designated source nodes. \\
%The node $i$ has a capacity of transmitting $c_i$ packets per slot, $\forall i \in V$. \\
For simplicity of exposition, we consider only static networks with a single source node $\texttt{r}$. However, the algorithm and its analysis presented in this paper extend to time-varying dynamic networks with multiple source nodes in a straightforward manner. We will consider time-varying networks in our numerical simulations.
%The problem of broadcasting in time-varying networks is explored in our numerical simulations.
%These packets need to be routed to all nodes in the graph by spanning trees, chosen for each arriving packet. 
%In this paper we are concerned with the problem of efficiently disseminating these packets among all nodes in the network in a multi-hop fashion. 
\subsection{Wireless Transmission Model}
When a node $i\in V$ is scheduled for transmission, it can transmit any of its received packets at the rate of $c_i$ packets per slot to \emph{all} of its out-neighbors over its outgoing hyperedge. See Figure \ref{hyperedge}. Due to the wireless interference constraint, only a selected subset of nodes can feasibly transmit over the hyperedges simultaneously without causing collisions. The wireless channel is assumed to be error-free otherwise. The set of all feasible transmission schedules may be described concisely using the notion of a \emph{Conflict Graph} $\mathcal{C}(\mathcal{G})$. The set of vertices in the conflict graph is the same as the set of nodes in the network $V.$ There is an edge between two nodes in the conflict graph if and only if these two nodes \emph{cannot} transmit simultaneously without causing collision. \editing{Note that our \emph{node-centric} definition of conflict graphs is a little different from the traditional \emph{edge-centric} definition of conflict graph, which concerns point-to-point transmissions \cite{jain2005impact} \cite{CG_def}}.\\
As the simplest example of the interference model, consider a wireless network where each node transmits on a separate channel, causing no inter-node interference. Hence, any subset of nodes can transmit at the same slot, and the conflict graph does not contain any edges. For another example, consider a wireless network subject to primary interference constraints.  In this case, the edge $(i,j)$ is \emph{absent} in the conflict graph $\mathcal{C}(\mathcal{G})$ if and only if nodes $i$ and $j$ are \emph{not} in the transmission range of each other and their out-neighbor-sets are disjoint. 
%In particular, assuming the unit disk model for wireless transmissions, the edge $(i,j)$ is absent in the conflict graph $\mathcal{C}(\mathcal{G})$ if and only if there is no node in the intersection of transmission regions of the nodes $i$ and $j$. \\
The set of all feasible transmission schedules $\mathcal{M}$ consists of the set of all \emph{Independent Sets} in the conflict graph. \\
\editing{Note that the above definition of feasible schedules and conflict graph, does not allow \emph{any} collision in the network. The same assumption was also used in \cite{towsley2008rate}, where such schedules were called ``interference-free''. However, due to the point-to-multi-point nature of the wireless medium, it is possible (and sometimes benefitial) to consider schedules that allow some collisions, so that a transmitted packet may be correctly received only by a strict subset of neighbors. As it will be clear in what follows, it is straightforward to extend our algorithm to allow such general schedules, albeit at the expense of additional computational complexity. In order to present the main ideas in a simplified setting, in the following, we stick to the ``interference-free'' schedules, as defined above. }

\begin{figure}
\centering
\begin{overpic}[width=0.23\textwidth]{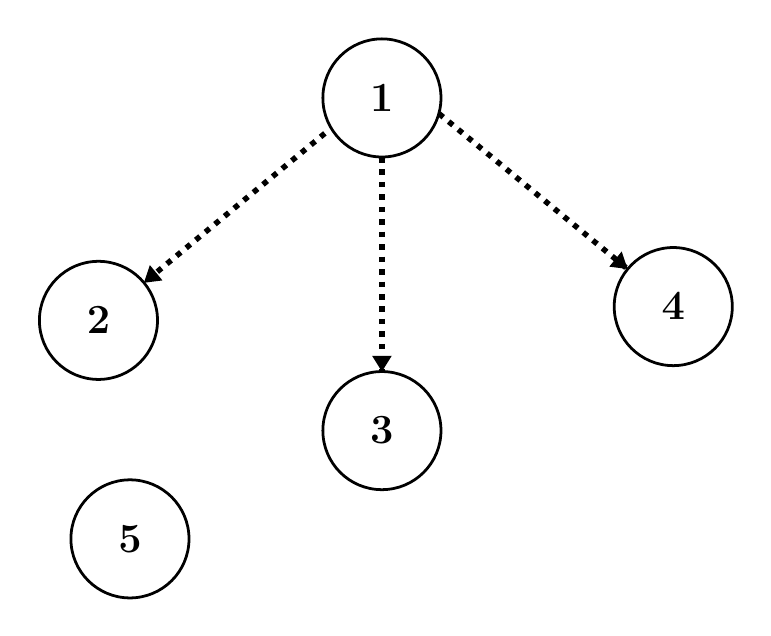}
%\put(56,50){\scriptsize{Capacity=2}}
%\put(89,32){\scriptsize{Capacity=1}}
\end{overpic}
%\vspace{5pt}
\caption{\small{An example of packet transmission over hyperedges - when the node $1$ transmits a packet, assuming no interference, it is received simultaneously by the neighboring nodes $2,3$ and $4$.}}
\label{hyperedge}
\end{figure}

\subsection{The Broadcast Policy-Space $\Pi$}
We first recall the definition of a connected dominating set of a graph $\mathcal{G}$ \cite{west2001introduction}. 
\begin{framed}
\begin{definition}[Connected Dominating Set] 
A connected dominating set $D$ of a graph $\mathcal{G}(V,E)$ is a subset of vertices with the following properties:
\begin{itemize}
\item The source node $\texttt{r}$ is in $D$.
\item The induced subgraph $\mathcal{G}(D)$ is connected.
\item Every vertex in the graph either belongs to the set $D$ or is adjacent to a vertex in the set $D$.
\end{itemize}
\end{definition}
\end{framed}
A connected dominating set $D$ is called \textbf{minimal} if $D\setminus \{v\}$ is not a connected dominating set for any $v \in D$. \editing{The set of all minimal connected dominating set is denoted by $\mathcal{D}$.}\\
% In a multihop network we have the following obvious precedence constraint on packet forwarding:
% \begin{framed}
% \begin{definition}[Precedence Constraint]
% A packet $p$ can be transmitted by a node $i$ at time $t$ only if it has been received by the node $i$ \emph{prior} to time $t$.
% \end{definition}
% \end{framed}
%It is clear that in order to broadcast a packet $p$ among all nodes in the network, the packet must be transmitted eventually by every node in some minimal connected dominating set $S \subset V$ of the graph $\mathcal{G}$. 
A packet $p$ is said to have been \emph{broadcasted} by time $t$ if the packet $p$ is present at every node in the network by time $t$.\\
It is evident that a packet $p$ is broadcasted if it has been transmitted sequentially by every node in a connected dominating set $D$. 
%In other words, a feasible sequence of transmission actions upon any broadcasted packet $p$  is in one-to-one correspondence with a minimal connected dominating set of the graph $\mathcal{G}$. 
An admissible broadcast policy $\pi$ is a sequence of actions $\{\pi_t\}_{t \geq 0}$ executed at every slot $t$. The action at time slot $t$ consists of the following three operations:
\begin{enumerate}
\item \textbf{Route Selection:} Assign a connected dominating set $D \in \mathcal{D}$ to every incoming packet at the source $\texttt{r}$ for routing.
\item \textbf{Node Activation:} Activate a subset of nodes from the set of all feasible activations $\mathcal{M}$.
\item \textbf{Packet Scheduling:} Transmit packets from the activated nodes according to some scheduling policy.
\end{enumerate}
The set of all admissible broadcast policies is denoted by $\Pi$.  The actions executed at every slot may depend on any past or future packet arrival and control actions. 

Assume that under the action of the broadcast-policy $\pi$, the set of packets received by node $i$ at the end of slot $T$ is $N^\pi_i(T)$. Then the set of packets $B(T)$ received by all nodes, at the end of time $T$ is given by
\begin{eqnarray}
B^\pi(T)= \bigcap_{i \in V} N^\pi_i(T).
\end{eqnarray}
\subsection{Broadcast Capacity $\lambda^*$}
Let $R^{\pi}(T)=|B^{\pi}(T)|$ denote the number of packets delivered to all nodes in the network up to time $T$, under the action of an admissible policy $\pi$. Also assume that the external packets arrive at the source node with expected rate of $\lambda$ packets per slot. The policy $\pi$ is called a \emph{broadcast policy of rate $\lambda$} if
\begin{equation} \label{limit_res}
\lim_{T \to \infty} \frac{R^{\pi}(T)}{T}= \lambda, \hspace{5pt} \mathrm{w.p.} 1
\end{equation} \label{pol_def}
The broadcast capacity $\lambda^*$ of the network is defined as 
\begin{eqnarray}
\lambda^*= \sup_{\pi \in \Pi}\{\lambda: \pi \text{ is a broadcast policy of rate } \lambda\}
\end{eqnarray}

The \textsf{Wireless Broadcast} problem is defined as finding an admissible policy $\pi$ that achieves the Broadcast rate $\lambda^*$.  
\section{Hardness Results} \label{hardness}
%In this section we will show that the problem \textsc{Wireless Broadcast} is computationally hard in the sense, any policy that solves the \textsc{Wireless Broadcast} problem, must necessarily produce a solution to another NP-hard problem. 

% Talk about distinction between finite and infinite horizon results

Since a broadcast policy, as defined above, continues to be executed forever \editing{(compared to the finite termination property of standard algorithms)}, the usual notions of computational complexity theory do not directly apply in characterizing the complexity of these policies. Nevertheless, we show that the closely related problem of finite horizon broadcasting is NP-hard. Remarkably, this problem remains hard even if the node activation constraints are eliminated (i.e., all nodes can transmit packets at the same slot, which occurs \emph{e.g.}, when each node transmits over a distinct channel). Thus, the hardness of the problem arises from the combinatorial nature of distributing the packets among the nodes. This is in sharp contrast with the polynomially solvable \textsc{Wired Broadcast} problem where the broadcast nature of the wireless medium is absent and different outgoing edges from a node can transmit different packets at the same slot \cite{edmonds} \cite{sinha2015throughput} \cite{sinha2016throughput2}. 

Consider the following finite horizon problem called \textsf{\textbf{Wireless Broadcast}}, with the input parameters $\mathcal{G},M,T$.
%\begin{framed}[]
\begin{itemize}
\item \textbf{INSTANCE:} A Graph $\mathcal{G}(V,E)$ with capacities $C$ on the nodes. A set of $\mathcal{M}$ of packets with $|\mathcal{M}|=M$ at the source and a time horizon of $T$ slots. 
\item \textbf{QUESTION:} Is there a scheduling algorithm $\pi$ which routes all of these $M$ packets to all nodes in the network by time $T\geq 2$, i.e. $B^{\pi}(T)=\mathcal{M}$?
\end{itemize}
%\end{framed}

% %The average broadcast-capacity is defined by $\bar{\lambda}(T)=M/T$. 
% \end{framed}
% For a fixed (and large) value of $T$, the broadcast capacity $\lambda^*$ of the network is defined as the largest value of $M/T$ ($M$ is the variable here) so that the problem admits a YES solution. \\
% Clearly, for the problem to have an YES solution it is necessary that $\frac{M}{T} \leq C_{\texttt{r}}$, because the source $\texttt{r}$ can only inject at most $C_{\texttt{r}}$ new packets in the network per slot. Since the problem is Monotone (i.e., YES answer for $M_1$ $\implies$ YES for $M_2$ where $M_2 \leq M_1$, it follows that the broadcast capacity of the network may be obtained by doing $\textsc{Binary Search}$ in $\mathcal{O}(\log(C_{\texttt{r}}T))=\mathcal{O}(\log(C_{\texttt{r}})+\log(T))$ times, which is strongly polynomial in the input-parameters.  

We prove the following hardness result:

 \begin{framed}
 \begin{theorem} \label{NP_hard_result}
 $\mathsf{Wireless} $ $ \mathsf{Broadcast}$ is \textbf{NP}-complete.
 \end{theorem}
 \end{framed}
 \editing{Proof of Theorem \ref{NP_hard_result} is based on reduction from the the \textbf{NP}-complete problem \textsf{Monotone Not All Equal $3$-SAT} \cite{schaefer1978complexity} to the \textsf{Wireless Broadcast} problem. The complete proof of the Theorem is provided in Appendix \ref{hardness_proof}. }\\
Note that the problem for $T=1$ is trivial as only the out-neighbors of the source receive $\min(C,M)$ packets at the end of the first slot. The problem becomes non-trivial  for any $T\geq 2$. In our reduction, we show that the problem is hard even for $T=2$. This reduction technique may be extended in a straightforward fashion to show that the problem remains \textbf{NP}-complete for any fixed $T\geq 2$.\\ 
% \subsubsection*{Strengthening of the result:}
% As a corollary of the reduction used for the above hardness result, it follows that the problem \textsf{Wireless Broadcast} remains \textbf{NP}-complete even with the following restrictions:
% \begin{enumerate}
% \item The wireless transmissions are non-interfering.
% %, \emph{e.g.}, when each node transmits over a different channel (c.f. \cite{edmonds}). 
% \item The graph $\mathcal{G}(V,E)$ is a DAG (c.f. \cite{sinha2015throughput}, \cite{sinha2016throughput2}).
% \item Capacity of each node can take any one of two different values. 
% %\item There are at most two distinct transmission capacities among all nodes.
% \item The in-degree of each node is at most $3$.
% %\item The number of messages $M$ and the deadline $T$ are bounded by universal constants.
% \end{enumerate} 
%\paragraph{Discussion}
The above hardness result is in sharp contrast with the efficient solvability of the broadcast problem in the setting of point-to-point channels. In wired networks, the broadcast capacity can be achieved by routing packets using maximal edge-disjoint spanning trees, which can be efficiently computed using Edmonds' algorithm \cite{edmonds}. In a recent series of papers \cite{sinha2015throughput} \cite{sinha2016throughput}, we proposed efficient throughput-optimal algorithms for wireless DAG networks in the static and time-varying settings. In a follow-up paper \cite{sinha2016throughput2}, the above line of work was extended to networks with arbitrary topology. In contrast, Theorem \ref{NP_hard_result} and its corollary establishes that achieving the broadcast capacity in a wireless network with broadcast channel is intractable even for a simple network topology, such as a DAG. Also notice that this hardness result is inherently different from the hardness result of \cite{shah2011hardness}, where the difficulty stems from the hardness of max-weight node activations, which is an Independent Set problem. The above result should also be contrasted with the hardness of the minimum energy broadcast problem \cite{clementi2001complexity}.

%\subsection{Characterization of the Broadcast Capacity}
%characterize the broadcast capacity using a stationary randomized policy (TBD). 

%\input{overview}
\section{Throughput-Optimal Broadcast Policy for a Relaxed Network} \label{bcast_algorithms}
%  In this section we derive a throughput-optimal online dynamic wireless broadcast policy $\pi^*$. The policy has a Max-Weight flavor but is drastically different from the traditional Back-Pressure like policies. In particular, unlike Back-pressure policy which makes packet forwarding decision for each packet hop-by-hop, the policy $\pi^*$ dynamically fixes the routing of a particular packet at the source. Also, unlike backpressure like algorithms which bases its decisions based on physical queue lengths and may be viewed as a \emph{closed loop} control, the policy $\pi^*$ bases its decisions based on a virtual queue length and may be viewed partly as an \emph{open loop} control. However, although it is oblivious to the physical queues, it does make control decisions based on virtual queues which act as a simple \emph{proxy} for the physical queues.  The main advantage of the virtual queues is that they are essentially single-hop queues which make their analysis tractable as opposed to the physical queues. There are several other theoretical and practical advantages associated with this approach, which will be discussed subsequently. 
%  
%\subsection{Overview}
 In this section, we give a brief outline of the design of the proposed broadcast policy, which will be described in detail in the subsequent sections. At a high level, the proposed policy consists of two interdependent modules - a control policy for a \emph{precedence-relaxed} virtual network described below, and a control policy for the actual physical network, described in Section \ref{physical_queue_stability}. Although, from a practical point of view, we are ultimately interested in the optimal control policy for the physical network, as we will soon see, this control policy is intimately related to, and derived from the dynamics of the relaxed virtual network. The concept of a precedence relaxed virtual network was first introduced in our recent paper \cite{sinha2016optimal}.
%The policy will be developed in the following sections.  
 
 %\subsection{System Model}
%Consider a network $\mathcal{G}(V,E)$, whose edge $e$ has capacity $c_e$ while active. \\
%The algorithm that we will be describing here extends readily to inter-mix of unicast, broadcast and multicast traffic with multiple sources. Virtually all network flow problems (e.g., Unicast, Broadcast, Multicast and Anycast) can be solved by this single algorithm.
 \subsection{Virtual Network and Virtual Queues}
In this section we  define and analyze the dynamics of an auxiliary virtual queueing process $\{\bm{\tilde{Q}}(t)\}_{t \geq 0}$. Our throughput-optimal broadcast policy $\pi^*$ will be described in terms of the virtual queues. We emphasize that virtual queues are not physical entities and they do not contain any physical packet. They are constructed solely for the purpose of designing a throughput-optimal policy for the physical network, which depends only on the value of the virtual queue lengths. More interestingly, the designed virtual queues correspond to a fairly natural \emph{single-hop} relaxation of the \emph{multi-hop} physical network, as detailed below.

%The virtual queue is an $n$-dimensional vector of non-negative integers, maintained at the source node $\texttt{r}$. Its components are iteratively updated at every slot, according to a control policy \textsf{UMW}, which will be described subsequently. 
\subsubsection*{A Precedence-relaxed System}
Consider an incoming packet $p$ arriving at the source, which is to be broadcasted through a sequence of transmissions by nodes in a connected dominating set $D_p \in \mathcal{D}$. Appropriate choice of the set $D_p$ is a part of our policy and will be discussed shortly. In reality, the packet $p$ cannot be transmitted by a non-source node $v \in D_p$ at time $t$ if it has not already reached the node $v$ by the time $t$. This causality constraint is known as the \emph{precedence constraint} in the literature \cite{lenstra1978complexity}. We obtain the virtual queue process $\bm{\tilde{Q}}(t)$ by relaxing the precedence constraint, i.e., in the virtual queuing system, the packet $p$ is made available for transmission by all nodes in the set $D_p$ when the packet first arrives at the source. See Figure \ref{VQ_fig1} for an illustration.

\begin{figure} [!ht] 
%\hspace*{-0.2cm}
\centering
\vspace{5pt}
\hspace{-20pt}
\begin{overpic}[width=0.48\textwidth]{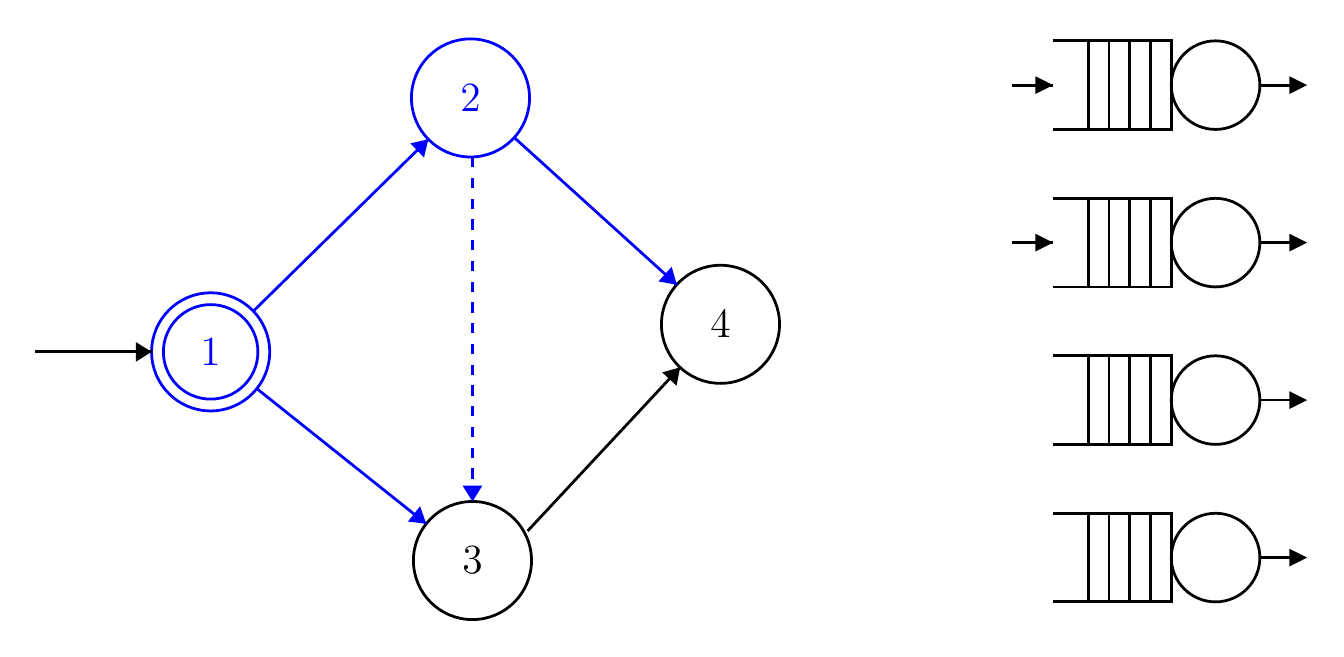}
%\put(11, 20.5){$\texttt{\textbf{\textcolor{blue}{r}}}$} 
%\put(50, 20.5){$c$} 
%\put(31.5, 33){$\textbf{\textcolor{blue}{a}}$} 
%\put(31.5, 7.5){$b$} 
\put(5, 25){$p$} 
\put(94, 47){$\scriptsize{\mu_{1}(t)}$} 
\put(94, 34){$\scriptsize{\mu_2(t)}$} 
\put(94, 22){$\scriptsize{\mu_3(t)}$} 
\put(94, 10.5){$\scriptsize{\mu_4(t)}$} 
\put(68.5, 35.7){$\tiny{\tilde{Q}_1(t)}$} 
\put(68.5, 24.5){$\scriptsize{\tilde{Q}_2(t)}$} 
\put(68.5, 13){$\scriptsize{\tilde{Q}_3(t)}$} 
\put(68.4, 2){$\scriptsize{\tilde{Q}_4(t)}$} 
\put(72.5, 42){$\textcolor{blue}{\scriptsize{\tilde{p}}}$} 
\put(72.5, 30){$\textcolor{blue}{\tilde{p}}$} 
%\put(78,47){$\textcolor{blue}{+1}$}
%\put(78,35.5){$\textcolor{blue}{+1}$}
\put(21, -8){\scriptsize{A Wireless Network $\mathcal{G}$}}
\put(70, -8){\scriptsize{Virtual Queues}}
\put(43, -1){{$D_p=\{1, 2\}$}}
\end{overpic}
\vspace{20pt}
\caption{\small{Illustration of the virtual queue system for the four-node wireless network $\mathcal{G}$. Upon arrival, the incoming packet $p$ is prescribed a connected dominating set $D_p=\{1,2\}$, which is used for its broadcasting. Relaxing the precedence constraint, packet $p$ is counted as an arrival to the virtual queues $\tilde{Q}_\texttt{1}$ and $\tilde{Q}_2$ at the \emph{same slot}. In the physical system, the packet $p$ may take a while before reaching node $2$, depending on the control policy.}}
\label{VQ_fig1}
\end{figure}
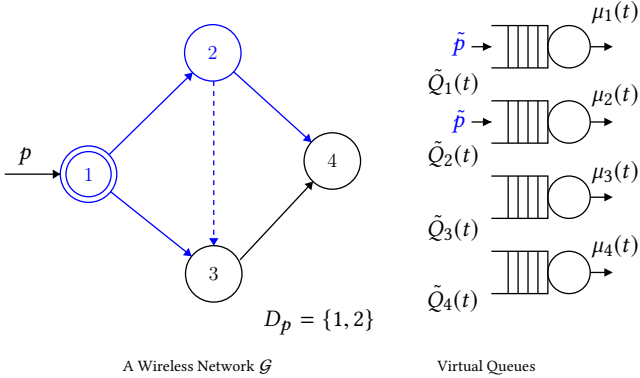

\subsubsection*{Dynamics of the Virtual Queues}

Formally, for each node $i \in V$, we define a virtual queue variable $\tilde{Q}_i(t)$. As described above, on the arrival of an external packet $p$ at the source $\texttt{r}$, the packet is replicated to a set of virtual queues $\{\tilde{Q}_i(t), i \in D_p\}$, where $D_p \in \mathcal{D}$ is a connected dominating set of the graph. Mathematically, this operation means that all virtual queue-counters in the set $D_p$ are incremented by the number of external arrivals at the slot $t$. We will use the control variable $A_i(t)$ to denote the number of packets that were routed to the virtual queue $\tilde{Q}_i$ at time $t$. The service rate $\bm{\mu}(t)$ allocated to the virtual queues is required to satisfy the same interference constraint as the physical network, i.e. $\bm{\mu}(t) \in \mathcal{M}, \forall t$.
%A feasible service rate of the virtual queues at the time $t$ is selected from the set $\mathcal{M}$. 
Hence, we can write the one step dynamics of the virtual queues as follows:
\begin{eqnarray} \label{queue-evolution}
\tilde{Q}_i(t+1) = (\tilde{Q}_i(t)+A_i(t)-\mu_i(t))^+, \hspace{10pt} \forall i \in V
\end{eqnarray}
\subsection{Dynamic Control of Virtual Queues}
In this section, we design a dynamic control policy to stabilize the virtual queues for all arrival rates $\lambda< \lambda^*$. This policy takes action (choosing the routes of the incoming packets and selecting a feasible transmission schedule) by observing the virtual queue-lengths only and, unlike popular unicast policies such as \emph{Backpressure}, does not require physical queue information. This control policy is obtained by minimizing one-step expected drift of an appropriately chosen Lyapunov function as described below. In the next section we will show how to combine this control policy for the virtual queues with an appropriate packet scheduling policy for the physical networks, so that the overall policy is throughput-optimal. \\
Consider the Lyapunov function $L(\cdot)$ defined as the Eucledian $2$-norm of the virtual queue lengths, i.e., 
 \begin{eqnarray}
  L(\bm{\tilde{Q}(t)})=  ||\bm{\tilde{Q}(t)}||_2= \sqrt{\sum_i \tilde{Q}_i^2(t)}
 \end{eqnarray}
 The one step drift $\Delta(t)$ of the Lyapunov function may be bounded as follows:
 
 \begin{eqnarray}  \label{drift_bound2}
  \Delta(t) &\equiv&   L(\bm{\tilde{Q}(t+1)})-L(\bm{\tilde{Q}(t)}) \nonumber \\
  %&=& ||\bm{\tilde{Q}}(t+1)||- ||\bm{\tilde{Q}}(t)||
  &=& \sqrt{\sum_i \tilde{Q}_i^2(t+1)}- \sqrt{\sum_i \tilde{Q}_i^2(t)}
%  &\stackrel{(a)}{=}& \frac{\sum_i \big(\tilde{Q}_i^2(t+1) - \tilde{Q}_i^2(t)\big)}{\sqrt{\sum_i \tilde{Q}_i^2(t+1)}+\sqrt{\sum_i \tilde{Q}_i^2(t)}} \nonumber \\
%  &\stackrel{(b)}{\leq}& \frac{1}{||\bm{\tilde{Q}}(t)||}\sum_i \big(\tilde{Q}_i^2(t+1) - \tilde{Q}_i^2(t)\big),\label{drift_bound}
 \end{eqnarray}
 To bound this quantity, notice that for any $x\geq 0$ and $y>0$, we have 
 \begin{eqnarray} \label{ineq1}
 	\sqrt{x}-\sqrt{y} \leq \frac{x-y}{2\sqrt{y}}
 \end{eqnarray}
 The inequality above follows by noting that RHS minus LHS is non-negative. 
% The inequality above simply follows by computing the difference of the RHS and LHS as follows:
% \begin{eqnarray*}
% 	\frac{x-y}{2\sqrt{y}}-(\sqrt{x}-\sqrt{y})&=& (\sqrt{x}-\sqrt{y})(\frac{\sqrt{x}+\sqrt{y}}{2\sqrt{y}}-1)\\
% 	&=& \frac{(\sqrt{x}-\sqrt{y})^2}{2\sqrt{y}} \geq 0
% \end{eqnarray*}
Substituting $x=||\bm{\tilde{Q}}(t+1)||^2$ and
$y=||\bm{\tilde{Q}}(t)||^2$ in the inequality \eqref{ineq1}, we have the following bound on the one-step drift \eqref{drift_bound2} for any $||\bm{\tilde{Q}}(t)||>0$
\begin{eqnarray} \label{drift_bound}
	\Delta(t) \leq \frac{1}{2||\bm{\tilde{Q}}(t)||}\bigg(\sum_i \big(\tilde{Q}_i^2(t+1)-\tilde{Q}_i^2(t) \big)\bigg)
\end{eqnarray}
%where, the equality $(a)$ follows after multiplying the numerator and denominator by $ \sqrt{\sum_i \tilde{Q}_i^2(t+1)}+\sqrt{\sum_i \tilde{Q}_i^2(t)}$ and the upper-bound $(b)$ follows by dropping the non-negative term $\sqrt{\sum_i \tilde{Q}_i^2(t+1)}$ from the denominator. 
%

% Consider the quadratic Lyapunov function $L(\bm{Q}(t))$ defined as a function of the queue-length of the virtual queues:
% \begin{eqnarray*}
% 	L(\bm{Q}(t))= \sum_{i \in V} Q_i^2(t) 
% \end{eqnarray*}
From the virtual queue dynamics \eqref{queue-evolution}, we have:
\begin{eqnarray*}
\tilde{Q}_i(t+1)^2 &\leq& (\tilde{Q}_i(t)-\mu_i(t)+A_i(t))^2 \\
&=& \tilde{Q}_i^2(t) + A_i^2(t)+\mu_i^2(t) +2\tilde{Q}_i(t)A_i(t)\\
&-& 2\tilde{Q}_i(t)\mu_i(t)-2\mu_i(t)A_i(t)
\end{eqnarray*}
Since $\mu_i(t) \geq 0$ and $A_i(t) \geq 0$, we have 
\begin{eqnarray} 
	\tilde{Q}_i^2(t+1) - \tilde{Q}_i^2(t) \leq A_i^2(t) + \mu_i^2(t) \nonumber \\
	+2\tilde{Q}_i(t)A_i(t)-2\tilde{Q}_i(t)\mu_i(t)\label{dyn_bound}
\end{eqnarray}
Hence, combining Eqns. \eqref{drift_bound} and \eqref{dyn_bound}, the one-step Lyapunov drift, conditional on the current virtual queue-length $\bm{\tilde{Q}}(t)$, under the action of an admissible policy $\pi$ is upper-bounded as: 
\begin{eqnarray}
	&&\mathbb{E}(\Delta^\pi(t)|\bm{\tilde{Q}}(t)=\bm{\tilde{Q}}) \nonumber \\
	&\stackrel{\mathrm{(def)}}{=}&\mathbb{E}\big(L(\bm{\tilde{Q}}(t+1))-L(\bm{\tilde{Q}}(t))|\bm{\tilde{Q}}(t)=\bm{\tilde{Q}}\big)\nonumber \\
	&\leq& \frac{1}{2||\bm{\tilde{Q}}||}\bigg(B + 2 \underbrace{\sum_{i}\tilde{Q}_i(t)\mathbb{E}\big(A_i^\pi(t)|\bm{\tilde{Q}}(t)=\bm{\tilde{Q}}\big)}_{(a)} \nonumber \\
	&-& 2 \underbrace{\sum_{i}\tilde{Q}_i(t)\mathbb{E}\big(\mu_i^\pi(t)|\bm{\tilde{Q}}(t)=\bm{\tilde{Q}}\big)}_{(b)} \bigg) \label{UMW_def}
\end{eqnarray} 
where the constant $B= \sum_{i}(\mathbb{E}A_i^2(t)+\mathbb{E}\mu_i^2(t))\leq n(\mathbb{E}A^2+c^2_{\max})$. 
%The expectation is taken with respect to the random arrivals and possible randomness in the policy $\pi$. \\
%The first expectation is with respect to the random arrival process and possible randomness in the policy $\pi$ and the second expectation is with respect to possible randomness in the policy $\pi$. \\
% Hence, by Foster-Lyapunov theorem, the virtual queue-length process $\{\bm{Q}(t)\}_{t \geq 0}$, which constitutes a Markov Chain under the policy \textsf{UMW}, is positive-recurrent. 
By minimizing the upper-bound on drift from Eqn. \eqref{UMW_def}, and exploiting the separable nature of the objective, we obtain the following control policy for the virtual queues:\\\\
\textbf{\textbf{U}niversal \textbf{M}ax \textbf{W}eight (\textsf{UMW}) policy for the Virtual Queues} \\

%\begin{framed}
 \textsc{\textbf{1. Route Selection}}: We minimize the term (a) in the above with respect to all feasible controls to obtain the following routing policy: Route the incoming packet at time $t$ along the minimum-weight connected dominating set (MCDS) $D^{\textrm{UMW}}(t)$, where the nodes are weighted by the virtual queue-lengths $\bm{\tilde{Q}}(t)$, i.e., 
 \begin{eqnarray}\label{routing_dec}
  D^{\textrm{UMW}}(t) = \arg\min_{D \in \mathcal{D}} \sum_{i \in V} \tilde{Q}_i(t) \mathds{1}(i\in \mathcal{D}) 
 \end{eqnarray} 
 \textsc{\textbf{2. Node Activations}}: We maximize the term (b) in the above with respect to all feasible controls to obtain the following node scheduling policy: At time $t$ activate a feasible schedule $\bm{\mu}^{\textrm{UMW}}(t)$ having the maximum weight, where the nodes are weighted by the virtual queue-lengths $\bm{\tilde{Q}}(t)$, i.e., 
 \begin{eqnarray} \label{sched_dec}
  M^{\textrm{UMW}}(t) = \arg\max_{M \in \mathcal{M}} \sum_{i \in V} \tilde{Q}_i(t) c_i \mathds{1}(i\in M) 
 \end{eqnarray}
%\end{framed}
%The pseudocode of the virtual queue management policy is described in Algorithm \ref{UMW_algo}.

%We consider the following activation and scheduling policy, which we call \textsf{UMW} (\textbf{U}niversal \textbf{M}ax-\textbf{W}eight Algorithm). This algorithm consists of two major components (a) Virtual Queue Management and (b) Link Scheduling and Routing of physical packets. 

%\begin{framed}
%\begin{itemize}
%\item \textbf{[Edge-Weight Assignment]} Assign each edge of a cost equal to $Q_e(t)$ at time slot $t$.
%\item \textbf{[Routing: MST] } Once a packet comes find a Minimum weight spanning tree $T$ according to these weights. Route the packet over the tree $T$.
%\item \textbf{[Scheduling: Max-Weight]} At a given time slot activate a subset of edges in the activation set $\mathcal{M}$ which maximizes the weight. 
%\begin{itemize}
%\item \textbf{[Virtual Queue Update: Arrival]} Increment the virtual queue counters by the number of packet arrivals on those queues corresponding to the edges of the tree $T$.
%\item \textbf{[Virtual Queue Update: Departure]} Decrement (up to zero) the counter of $Q_e(t)$ by $\mu_e(t)$. 
%\end{itemize}
%\end{itemize}
%\end{framed}
%The following theorem is fundamental in connection with the analysis of the \textsf{UMW} policy:
In connection with the virtual queue systems $\bm{\tilde{Q}}(t)$, we establish the following theorem which will be essential in the proof of the throughput-optimality of the overall algorithm involving physical queues. 
\begin{framed}
\begin{theorem} \label{VQ_stability}
For any arrival rate $\lambda < \lambda^*$ the virtual queue process $\{\bm{Q}(t)\}_{t \geq 0}$ is positive recurrent under the action of the \textsf{UMW} policy and 
\begin{equation*} \label{log_bdd}
 \max_i \tilde{Q}_i(t) = \mathcal{O}(\log t)\footnotemark, \hspace{10pt} \textrm{w.p.} \hspace{2pt} 1.
\end{equation*}
%\textbf{strongly-stable} under the action of the \textsf{UMW} policy, i.e., 
% \begin{eqnarray*}
%  \limsup_{T\to \infty}\frac{1}{T}\sum_{t=1}^{T} \sum_{i=1}^{n}\mathbb{E}Q_i(t)  \leq B, 
% \end{eqnarray*}
% for some finite constant $B$.
\end{theorem}
\end{framed}
\footnotetext{Recall that, $f(t)=\mathcal{O}(g(t))$ if there exist a positive constant $c$ and a finite time $t_0$ such that $f(t)\leq c g(t), \forall t\geq t_0$.}
The proof of Theorem \ref{VQ_stability} involves construction of an efficient randomized policy and using it with a sharper form of the Foster-Lyapunov theorem by Hajek \cite{hajek1982hitting}. This leads to the desired sample path result. The proof is provided in Appendix \ref{stability_proof_VQ}.
\paragraph*{Discussion of the Result}
Even though the virual queue process is positive recurrent under the action of the \textsf{UMW} policy, it is not true that they are uniformly bounded almost surely. Theorem \ref{VQ_stability} states that, instead, the virtual queue lengths increase at most logarithmically with time almost surely. Theorem \ref{VQ_stability} also strengthens the result of Theorem 2.8 of \cite{neely2010stochastic}, where an almost sure $o(t)$ bound was established for the queue lengths\footnote{We say $f(t)=o(g(t))$ if $\lim_{t \to \infty}\frac{f(t)}{g(t)}=0$.}. \\
In the rest of the paper, we will primarily focus on the typical sample paths $\mathcal{E}$ of the virtual queue process satisfying the above almost sure bound.  Formally, we define the set $\mathcal{E}$ to be
\begin{eqnarray} \label{q_bd}
 \max_i \tilde{Q}_i(\omega, t) = \mathcal{O}(\log(t)), \hspace{10pt} \forall \omega \in \mathcal{E},
\end{eqnarray}
where $\mathbb{P}(\mathcal{E})=1$ from Theorem \ref{VQ_stability}.

\subsection{Bounds on the Virtual Queue }
%Denote the total number of packet arrivals to the virtual queue $\tilde{Q}_i$ within the time interval $(\tau, t], \tau \leq t$ by $\tilde{A}_i(\tau, t)$. 
Recall that the random variable $A_i(t)$ denotes the total number of packets injected to the virtual queue $\tilde{Q}_i$ at time $t$. Similarly, the random variable $\mu_i(t)$ denotes the service rate from the virtual queue $\tilde{Q}_i$ at time $t$. Hence, the total number of packets that have been injected into any virtual queue $\tilde{Q}_i$ within the time interval $[t_1, t_2)$, $t_1 \leq t_2$ is given by 
\begin{eqnarray} \label{arr_vq}
 A_i(t_1, t_2) = \sum_{\tau=t_1}^{t_2-1} A_i(\tau).
\end{eqnarray}
Similarly, the total amount of service offered to the virtual queue $\tilde{Q}_i$ within the time interval $[t_1, t_2)$ is given by 
\begin{eqnarray} \label{ser_vq}
 S_i(t_1,t_2)= \sum_{\tau=t_1}^{t_2-1} \mu_i(\tau).
\end{eqnarray}
Using the well-known Skorokhod representation theorem \cite{kleinrock1975queueing} of the Queueing recursion \eqref{queue-evolution}, we have \footnote{Note that, for simplicity of notation and without any loss of generality, we have assumed the system to be empty at time $t=0$.}
\begin{eqnarray} \label{skorokhod_map}
	\tilde{Q}_i(t)= \sup_{1\leq \tau \leq t} \big(A_i(\tau,t)-S_i(\tau, t)\big)^+.
\end{eqnarray}
Since the virtual queues $\bm{\tilde{Q}}$ are controlled by the \textsf{UMW} policy, combining Eqn. \eqref{q_bd} with \eqref{skorokhod_map}, we have for all typical sample paths $\omega \in \mathcal{E}$:
\begin{eqnarray} \label{arr_bd}
 A_i(\omega; \tau,t) \leq S_i(\omega; \tau, t) + F(\omega, t), \hspace{10pt} \forall  \tau \leq t, i \in V,
\end{eqnarray}
where $F(\omega, t)= \mathcal{O}(\log t)$. 
In other words, equation \eqref{arr_bd} states that under the \textsf{UMW} policy, for any packet arrival rate $\lambda < \lambda^*$,  the total number of packets that are routed to any virtual queue $\tilde{Q}_i$ may exceed the total amount of service offered to that queue in any time interval $[\tau, t)$ by at most an additive term of $\mathcal{O}(\log t)$ almost surely. In the following section, we will show that this arrival condition enables us to design a throughput-optimal broadcast policy.

%\subsection{Broadcast Policy for the Physical Network}
%Finally, we describe 

%\end{proof}
\section{Control of the Physical Network} \label{physical_queue_stability}
%\section{Optimal Control of the Physical Network} \label{physical_queues}
%Having determined the routing and link activation policy in section \ref{virtual_queues}, in this section we design the third and final component of the \textbf{UMW} policy, namely packet scheduling for the multi-hop physical queues. Since an active edge may transmit only one packet per slot, the packet scheduler should efficiently resolve contention when multiple packets want to cross an edge $e$ at the same time-slot $t$. \\
With the help of the one-hop virtual queue structure designed in the previous section, we now focus our attention on designing a throughput-optimal control policy for the multi-hop physical network. Recall from Section \ref{model} that a broadcast policy for the physical network is specified by the following three components: \textbf{(1)} Route Selection, \textbf{(2)} Node Activation, and \textbf{(3)} Packet Scheduling. In our proposed broadcast policy, components (1) and (2) for the physical network are identical to the corresponding components in the virtual network. In other words, an incoming packet $p$ at time $t$ is prescribed a route (i.e., a connected dominating set) given by Eqn.\ \eqref{routing_dec} and the set of nodes given by Eqn.\ \eqref{sched_dec} are scheduled for transmission in that slot. Note that, both these decisions are based on the instantaneous virtual queue lengths $\bm{\tilde{Q}}(t)$. In particular, it is possible that a particular node, with positive virtual queue length, is scheduled for transmission in a slot, even though it does not have any packets to transmit in its physical queue. The surprising fact, that will follow from Theorem \ref{UMW-optimality}, is that this kind of wasted transmissions are rare and \emph{do not affect throughput}.     \\
% As a direct consequence of the above identification, we conclude that the number of packets, which arrive at the source in the time interval $[\tau, t)$ and whose prescribed route contains the node $i$, is exactly equal to the corresponding arrival in the virtual network $A_i(\tau, t)$, given by Eqn. \eqref{arr_vq}. Similarly, total service offered by the physical node $i$ in the time interval $(\tau, t]$ is exactly given by $S_i(\tau, t)$, defined in Eqn. \eqref{ser_vq}. Thus, the bound in Eqn. \eqref{arr_bd} may be interpreted in terms of the packets arriving in the physical network. \\
\textbf{Packet Scheduling:} There are many possibilities for the component \textbf{(3)}, i.e. Packet scheduling in the physical network. Recall that, the packet scheduling component selects packet(s) to be transmitted (subject to the node capacity constraint) when multiple packets contend for transmission by an active node and plays a role in determining the physical queuing process. In this paper, we consider a priority based scheduler which gives priority to the packet which has been transmitted by the nodes \emph{the least number of times}. We call this scheduling policy \textbf{L}east \textbf{T}ransmitted \textbf{F}irst or \textsf{LTF}. The \textsf{LTF} policy is inspired from the \emph{Nearest To Origin} policy of Gamarnik \cite{Gamarnik}, where it was shown to stabilize the queues for the unicast problem in wired networks in a deterministic adversarial setting. In spite of the high level similarities, however, we emphasize that these two policies are different, as the \textsf{LTF} policy works in the broadcast setting with point-to-multi-point transmissions and involves packet duplications.

%The complete description of the broadcast policy is given in Algorithm \ref{UMW_algo}. 
\begin{framed}
 \begin{definition}[The policy \textsf{LTF}]
  If multiple packets are available for transmission by an active node at the same time slot $t$, the \textsf{LTF} scheduling policy gives priority to a packet which has been transmitted the smallest number of times among all other contending packets. 
 \end{definition}
\end{framed}

See Figure \ref{LTF_fig} for an illustration of the \textsf{LTF} policy. 
%A complete pseudo code for the \textsf{UMW} policy is provided in Algorithm \ref{UMW_algo}. 
 
\begin{figure}
\includegraphics[width=0.45\textwidth]{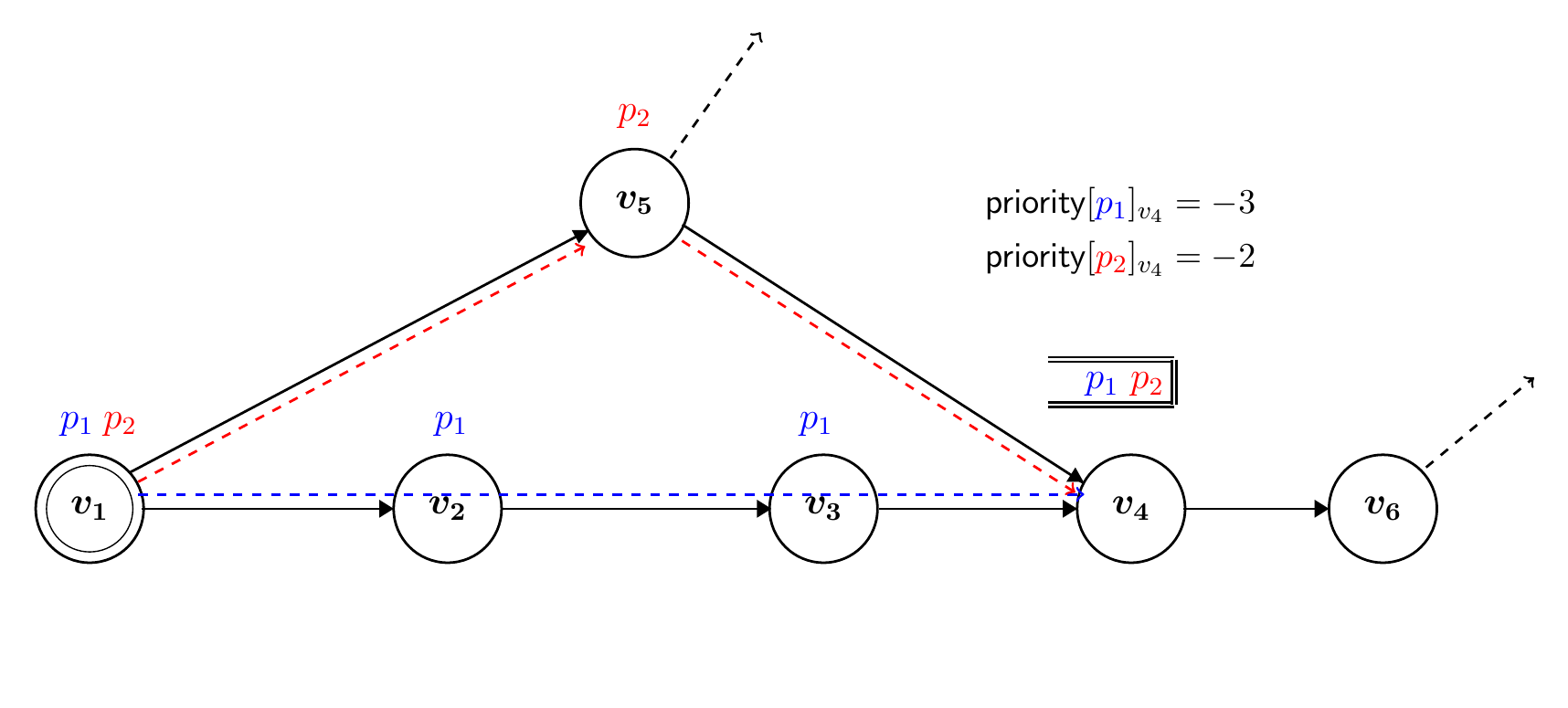}
\caption{\small{A schematic diagram depicting the scheduling policy LTF in action. The packet $p_1$'s broadcast route consists of the nodes $\{v_1,v_2,v_3,v_4, \ldots\}$ and the packet $p_2$'s broadcast route consists of the nodes $\{v_1,v_5, v_4, \ldots\}$ as shown in the figure. At node $v_4$, according to the LTF policy, the packet $p_2$ has higher priority than the packet $p_1$ for transmission. }}
%A schematic diagram showing the scheduling policy LTF in action. The packets $p_1$ and $p_2$ originate from the sources $S_1$ and $S_2$. Part of their assigned routes are shown in blue and red respectively. The packets contend for crossing the active edge $e_3$ at the same time slot. According to the ENTO policy, the packet $p_2$ has higher priority (having crossed a single edge $e_4$ from its source) than $p_1$ (having crossed two edges $e_1$ and $e_2$ from its source) for crossing the edge $e_3$. Note that, although a copy of $p_1$ might have already crossed the edge $e_5$, this edge does not fall in the path connecting the source $S_1$ to the edge $e_3$ and hence does not enter into priority calculations. }}
\label{LTF_fig}
\end{figure}

\subsection{Stability of the Physical Queues}
Let us denote the length of the physical queue at node $i$ at time $t$ by $Q_i(t)$. Note that the number of packets which arrive at the source in the time interval $[\tau, t)$ and whose prescribed route contains the node $i$, is equal to the corresponding arrival in the virtual network $A_i(\tau, t)$, given by Eqn.\ \eqref{arr_vq}. Similarly, total service offered by the physical node $i$ in the time interval $(\tau, t]$ is given by $S_i(\tau, t)$, defined in Eqn.\ \eqref{ser_vq}. Thus, the bound in Eqn.\ \eqref{arr_bd} may be interpreted in terms of the packets arriving to the physical network. This leads to the following theorem: %which establishes the stability of the physical queues under the \textsf{UMW} policy:
\begin{framed}
 \begin{theorem} \label{physical_stability-theorem}
  Under the action of the \textsf{UMW} policy with \textsf{LTF} packet scheduling, we have for any arrival rate $\lambda < \lambda^*$,
  \begin{eqnarray*}
   \sum_{i \in V} Q_i(t) = \mathcal{O}(\log t),  \hspace{10pt} \mathrm{w.p.} \hspace{1pt} 1.
  \end{eqnarray*}
 This implies that,  
  \begin{eqnarray*}
   \lim_{t \to \infty} \frac{\sum_{i \in V}Q_i(t)}{t} =0 , \hspace{10pt} \mathrm{w.p.} \hspace{1pt} 1, 
  \end{eqnarray*}
  i.e., the physical queues are rate-stable \cite{neely2010stochastic}.
\end{theorem}
\end{framed}
%\begin{proof}
Theorem \ref{physical_stability-theorem} is established by combining the key sample path property of arrivals and departures from Eqn.\ \eqref{arr_bd}, with an adversarial queueing theoretic argument of Gamarnik \cite{Gamarnik}. The complete proof of the Theorem is provided in Appendix \ref{physical_stability-proof}.\\  
%\end{proof}
%\vspace{-10pt}
As a direct consequence of Theorem \ref{physical_stability-theorem}, we have the main result of this paper:
\begin{framed}
\begin{theorem} \label{UMW-optimality}
 \textsf{UMW} is a throughput-optimal wireless broadcast policy.
\end{theorem}
\end{framed}
\begin{proof}
%For any class $c \in \mathcal{C}$, the number of packets $R^{(c)}(t)$, received by all nodes $ i \in \mathcal{D}^{(c)}$ may be bounded as follows:
The total number of packets $R(t)$, received by all nodes in common up to time $t$ may be bounded in terms of the physical queue lengths as follows
\begin{eqnarray} \label{bd_eqn1}
  A(0,t) - \sum_{i \in V} {Q}_i(t) \stackrel{(*)}{\leq} R(t) \leq A(0,t),
\end{eqnarray}
where the lower-bound $(*)$ follows from the simple observation that if a packet $p$  has not reached at all nodes in the network, then at least one copy of it must be present in some physical queue. \\ 
Dividing both sides of Eqn.\ \eqref{bd_eqn1} by $t$, taking limits and using the Strong Law of Large Numbers and Theorem \ref{physical_stability-theorem}, we conclude that 
\begin{eqnarray*}
 \lim_{t \to \infty} \frac{R(t)}{t} = \lambda, \hspace{5pt} \textrm{w.p.} 1.
\end{eqnarray*}
Hence, from the definition \eqref{pol_def}, we conclude that \textsf{UMW} is throughput-optimal.
\end{proof}
\subsection*{Efficient Implementation}
It is evident from the description of the \textsf{UMW} policy that the routing and node activation decisions are made using the virtual queue lengths $\tilde{\bm{Q}}(t)$, whereas the physical packet scheduling decisions are based on the contents of the physical queues at each node. In the following, we discuss efficient implementation of each of the three components in detail.
\subsubsection{Routing}
A broadcast route (MCDS) is computed for each packet immediately upon its arrival according to Eqn.\ \eqref{routing_dec}, and copied into its header field. The route selection involves solving an MCDS problem with the nodes weighted by the corresponding virtual queue lengths, which is \textbf{NP}-hard \cite{michael1979computers}. This is consistent with the hardness of the wireless broadcast problem, proved in Theorem \ref{NP_hard_result}. Assuming bi-directional wireless links, a polynomial time  $\mathcal{O}(\log n)$ approximation algorithm for this problem is available for general graphs \cite{guha1998approximation}. Furthermore, constant factor approximation algorithms for this problem are available for unit disk graphs \cite{thai2006connected}. 

%It can be readily shown that using this approximation algorithm instead, one can achieve $\mathcal{O}(\log n)$ fraction of the broadcast capacity. 

\subsubsection{Node Activation}
At every slot a non-interfering subset of nodes is activated by choosing a maximum weight independent set in the conflict graph $\mathcal{C}(\mathcal{G})$, where the nodes are weighted by their corresponding virtual queue lengths, see Eqn.\ \eqref{sched_dec}. The problem of finding a maximum weight independent set in a general graph is known to be \textbf{NP}-hard \cite{michael1979computers}. However, for the special case, such as unit disk graphs, constant factor approximation algorithms are available \cite{das2015approximation}. Note that, the same issue arises in the classical max-weight policies \cite{tassiulas}. \\ 
\editing{By a similar analysis, it can be shown that using an $\alpha \geq 1$ approximation algorithm for routing and $\beta \geq 1$ approximation algorithm for node activation, we can achieve $\frac{1}{\max{(\alpha, \beta)}}$ fraction of the optimal broadcast capacity of the network.}

\subsubsection{Packet Scheduling}
The LTF policy can be efficiently implemented by maintaining a \emph{min-heap} data-structure per node. The initial priority of each incoming packet at the source is set to zero. Once a packet $p$ is received at a node $i$ and the node $i$ is included in its list of required transmitting node, its priority is decreased by one and it is inserted to the min-heap maintained at node $i$. Naturally, a node simply discards multiple receptions of the same packet.\\

\section{Simulation Results} \label{simulation}
\subsection{Interference-free Network}
As a proof of concept, we first simulate the \textbf{UMW} policy in a simple wireless network with known broadcast capacity. Consider the network shown in Figure \ref{sim_net}. Here node $1$ is the source having a transmission capacity $C_1=2$. All other nodes in the network have  unit transmission capacity. Assume that the channels are non-interfering, i.e., all nodes can transmit in a slot (this holds, e.g.,  if the nodes transmit on different frequencies). Since the broadcast capacity of any wireless network is upper-bounded by the capacity of the source, we readily have $\lambda^* \leq 2$. Also, it can be seen from Figure \ref{sim_net} that by transmitting the even numbered packets from nodes $2$ and $5$ (shown in blue) and the odd numbered packets from nodes $3$ and $4$, a broadcast rate of $2$ packets per slot can be achieved. Hence, the broadcast capacity of the network is $\lambda^*=2$. 
\begin{figure}
\centering
\begin{overpic}[width=0.35\textwidth]{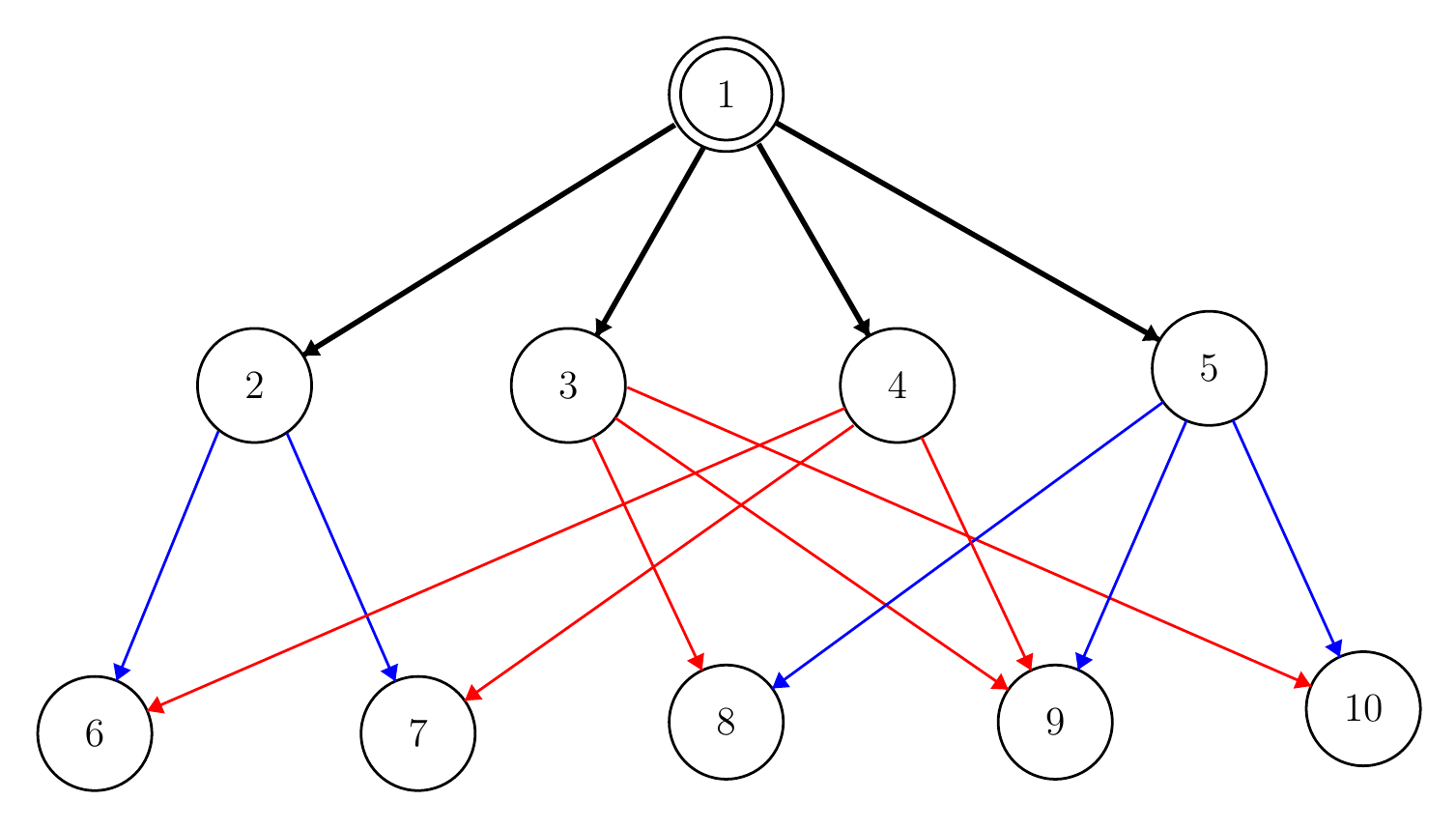}
\put(56,50){\scriptsize{Capacity=2}}
\put(89,32){\scriptsize{Capacity=1}}
\end{overpic}
%\vspace{5pt}
\caption{\small{A wireless network with non-interfering channels. The broadcast capacity of the network is $\lambda^*=2$.}}
\label{sim_net}
\end{figure}
\begin{figure}
\centering
\begin{overpic}[width=0.28\textwidth]{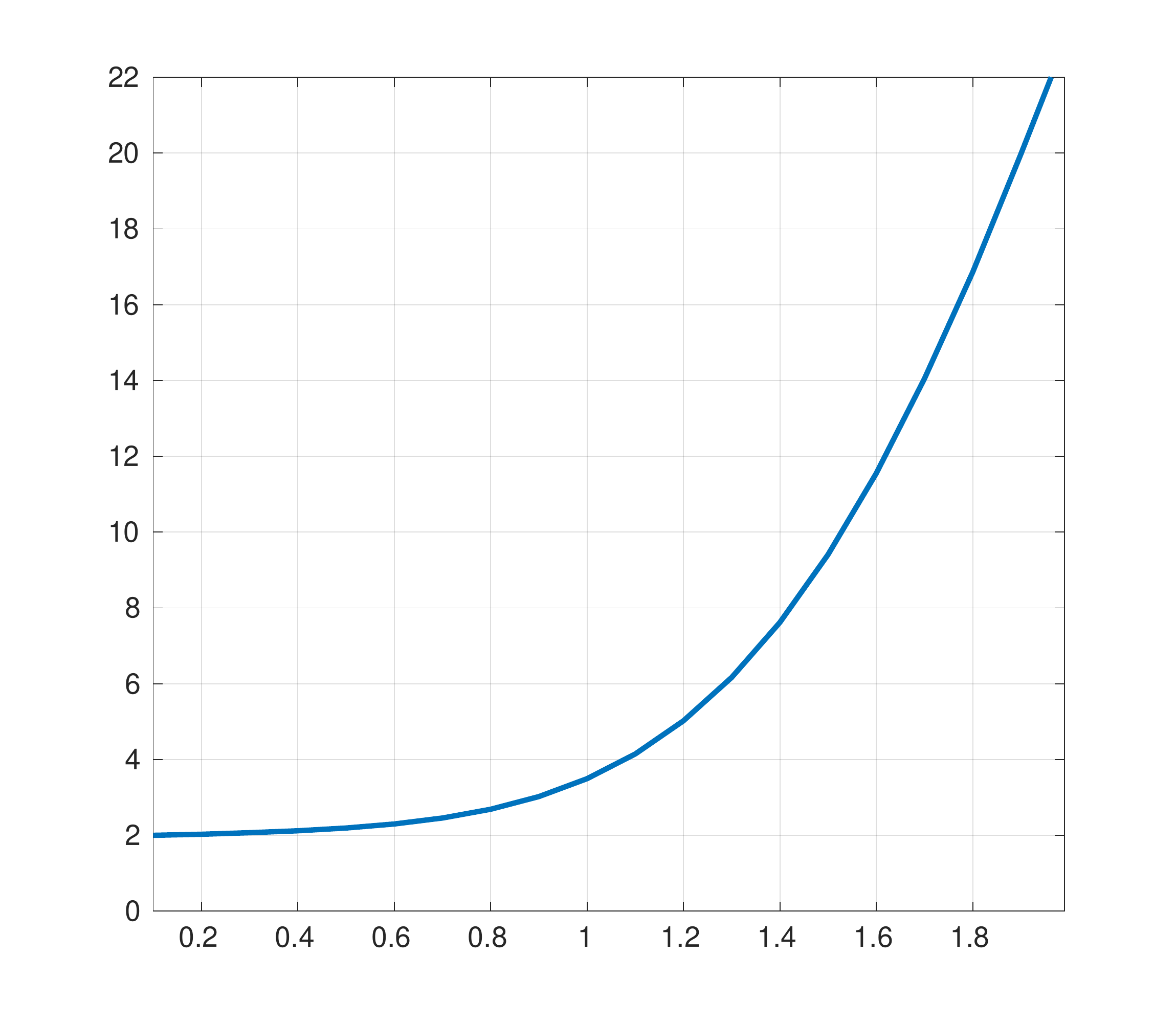}
\put(38,0){\small{Arrival rate $\lambda$}}
\put(2,24){\rotatebox{90}{\scriptsize{Broadcast Delay}}}
 \end{overpic}
%\vspace{5pt}
\caption{\small{Plot of the broadcast delay incurred by the UMW policy as a function of the arrival rate $\lambda$ in the network shown in Fig. \ref{sim_net}.}}
\label{rate_fig1}
\end{figure}
Figure \ref{rate_fig1} shows the average broadcast delay with the packet arrival rate $\lambda$  in this network under the action of the proposed \textsf{UMW} policy. Note that the minimum delay is at least $2$ as it takes at least two slots for any arriving packet to reach the nodes in the third layer. The plot confirms that the dynamic policy achieves the full broadcast capacity. 
\subsection{Network with Interference Constraints}
Consider the $3 \times 3$ wireless grid network, shown in Fig. \ref{grid_net}. Assume that the transmissions are limited by primary interference constraints, i.e, two nodes cannot transmit together if the transmissions interfere at any node in the network. Assume that any node, if activated, has a transmission rate of one packet per slot. In this setting we have the following upper-bound on the broadcast capacity of the network.

\begin{figure}
%\hspace*{-0.2cm}
\centering
\begin{overpic}[width=0.2\textwidth]{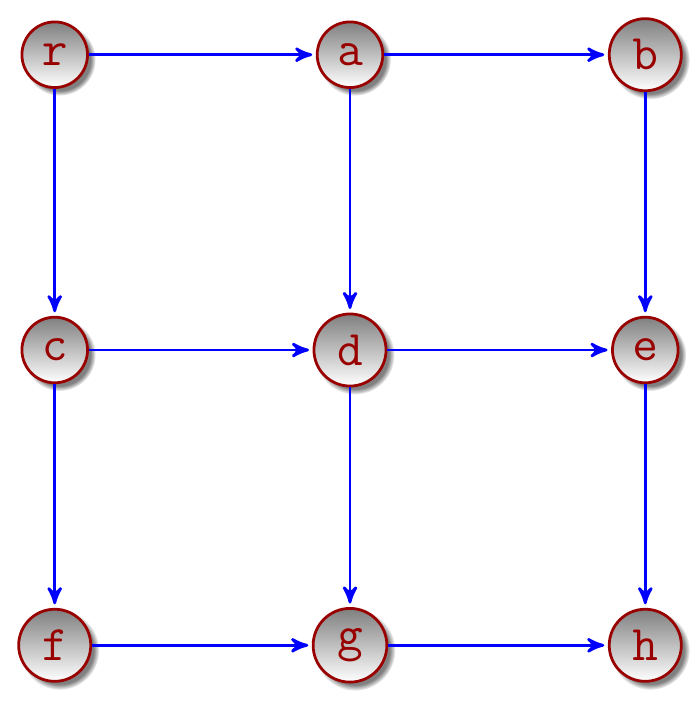}
\put(-8,100){\small{Source}}
 \end{overpic}
\vspace{5pt}
\caption{\small{A $3 \times 3$ wireless grid network with primary interference constraints. The wireless broadcast capacity ($\lambda^*$) of the network is at most $\frac{1}{3}$.}}
%\end{overpic}
\label{grid_net}
\end{figure}
\begin{framed}

\begin{lemma} \label{cap_bd1}
 The broadcast capacity of the $3\times 3$ grid network is at most $\frac{1}{3}$.
\end{lemma}
\end{framed}
The proof of the lemma is provided in Appendix \ref{cap_bd_proof}. \\
In Figure \ref{grid_bcast_fig} we show the broadcast delay as a function of the packet arrival rate, under the action of the \textsf{UMW} policy on the right most curve marked (a). From the plot, we observe that the delay-throughput curve has a vertical asymptote approximately along the straightline $\lambda=\frac{1}{3}$. This, together with lemma \ref{cap_bd1}, immediately implies that the broadcast capacity of the network is $\lambda^*=\frac{1}{3}$ and confirms the throughput-optimality of the \textsf{UMW} policy.  
% \begin{figure}
% \centering
% \begin{minipage}{0.5\textwidth}
%  \begin{overpic}[width=0.22\textwidth]{../figures/grid_network}
% \put(-8,100){\small{Source}}
%  \end{overpic}
% \vspace{5pt}
% \caption{\small{A $3 \times 3$ wireless grid network with primary interference constraints. The broadcast capacity of the network is $\lambda^*\leq \frac{1}{3}$.}}
% %\end{overpic}
% \label{grid_net}
% \end{minipage}
% \begin{minipage}{0.5\textwidth}
% \begin{overpic}[width=0.3\textwidth]{../figures/grid_delay}
% \put(38,0){\small{Arrival rate $\lambda$}}
% \put(2,24){\rotatebox{90}{\scriptsize{Broadcast Delay}}}
%  \end{overpic}
% %\vspace{5pt}
% \caption{\small{Plot of the broadcast delay incurred by the UMW policy as a function of the arrival rate $\lambda$ in the $3\times 3$ wireless grid network shown in Fig. \ref{grid_net}.}}
% \label{grid_bcast_fig}
% \end{minipage}
% \end{figure}

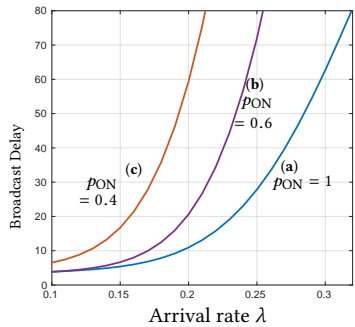
\begin{figure}
\centering
\begin{overpic}[width=0.29\textwidth]{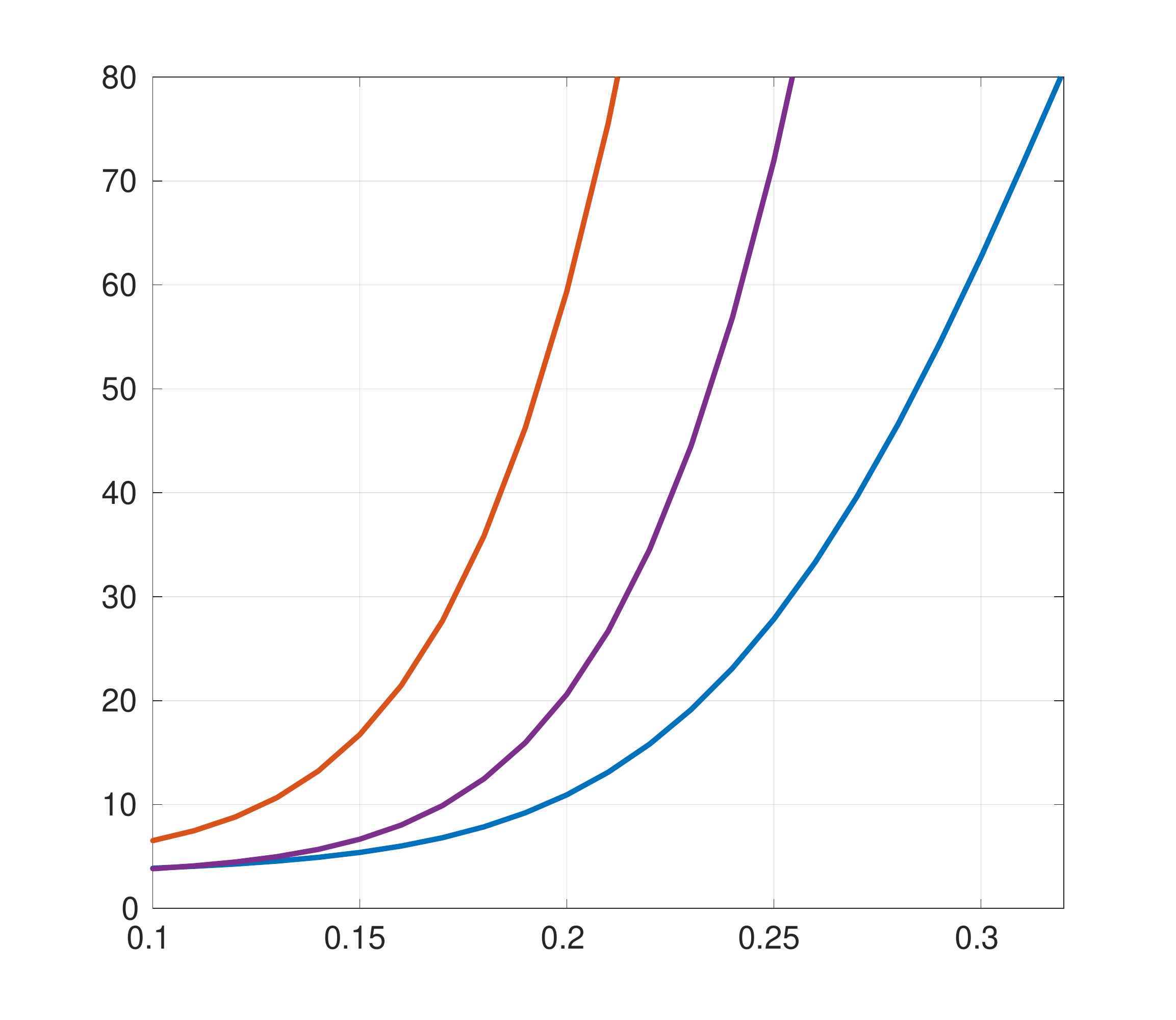}
\put(38,0){\small{Arrival rate $\lambda$}}
\put(70,35){\scriptsize{$p_{\textrm{ON}}=1$}}
\put(72,39){\scriptsize{(\textbf{a})}}
\put(62,56){\scriptsize{$p_{\textrm{ON}}$}}
\put(60,50){\scriptsize{$=0.6$}}
\put(63,60){\scriptsize{(\textbf{b})}}
\put(32,38){\scriptsize{(\textbf{c})}}
\put(22,35){\scriptsize{$p_{\textrm{ON}}$}}
\put(20,30){\scriptsize{$=0.4$}}
\put(2,24){\rotatebox{90}{\scriptsize{Broadcast Delay}}}
 \end{overpic}
%\vspace{5pt}
\caption{\small{Plot of the broadcast delay incurred by the UMW policy as a function of the arrival rate $\lambda$ in the $3\times 3$ wireless grid network shown in Fig. \ref{grid_net}.}}
\label{grid_bcast_fig}
\end{figure}
\subsubsection*{Broadcasting in a Time Varying Network}
Next, we simulate the \textsf{UMW} broadcast policy on a time-varying wireless grid network of Figure \ref{grid_net}, in which the nodes are not always available for transmission (\emph{e.g.,} they are sensors in sleep mode). In particular, we assume a simplified model where each node is active for potential transmission at a slot independently with some fixed but unknown probability $p_{\textrm{ON}}$.  The delay performance of the proposed \textsf{UMW} broadcast policy is shown in Figure \ref{grid_bcast_fig} (b) and (c) for two cases, $p_{\textrm{ON}}=0.6$ and $p_{\textrm{ON}}=0.4$ respectively. Following similar analysis as in the preceding sections, it can be shown that the \textsf{UMW} policy is also throughput-optimal for time-varying networks. Hence, from the plot it follows that the broadcast capacities of the time-varying $3\times 3$ wireless grid network are $\approx 0.26$ and $\approx 0.22$ packets per slot, for the activity parameter $p_{\textrm{ON}}=0.6$ and $p_{\textrm{ON}}=0.4$ respectively. 
%We consider the case when $p_i=0.8, \forall i$. The performance of the \textsf{UMW} policy is shown in the plot.   

\section{Conclusion} \label{conclusion}
In this paper we obtained the first throughput-optimal broadcast policy for wireless networks with point-to-multi-point links and arbitrary scheduling constraints. The policy is derived using the powerful framework of \emph{precedence-relaxed virtual network}, which we used earlier for designing throughput-optimal policies for networks with point-to-point links. Packet routing and scheduling decisions are made by solving standard optimization problems on the network, weighted by the virtual queue lengths. The policy is proved to be throughput optimal by a combination of Lyapunov method and a sample path argument using adversarial queueing theory. Extensive simulation results demonstrate the efficiency of the proposed policy in both static and dynamic network settings. There exist several interesting directions to extend this work. First, in our simplified model, we assumed that interference-free wireless transmissions are also error-free. A more accurate wireless channel model would incorporate the possibility of packet losses associated with each individual receiving nodes, due to fading and receiver noise \cite{sinha2016throughput}. Second, it remains unknown whether the \textsf{UMW} policy is still throughput optimal if the routing and node activations are made using the corresponding \emph{physical queue} lengths as compared to the virtual queues. A positive result in this direction would lead to a more efficient implementation. 
\bibliographystyle{ACM-Reference-Format}
\bibliography{MIT_broadcast_bibliography}
%\newpage 
\section{Appendix} \label{appendix}
\subsection{Proof of Hardness of the \large{\textsc{Wireless Broadcast}} Problem} \label{hardness_proof}

We start with the following lemma 
\begin{lemma}
 \textsf{Wireless Broadcast} is in \textsf{NP}.
\end{lemma}

\begin{proof}
From the formulation, the problem \textsf{Wireless Broadcast} is a Decision Problem. Also, if it has a \textsf{Yes} answer then there is a scheduling algorithm which serves as a certificate. Hence the problem belongs to \textsf{NP}. 
\end{proof}

Next we show that an \textbf{NP}-complete problem, named \textsc{Monotone Not All Equal $3$-SAT (MNAE-$3$SAT)} reduces to the problem \textsc{Wireless Broadcast} in polynomial time. This will complete the reduction. \\
We begin with the description of the problem MNAE-$3$SAT:
%In the following we describe the problem \textsc{MNAE-3SAT}.  
\subsubsection{The problem \textsc{MNAE-3SAT} }
\begin{framed}
\begin{itemize}
\item \textbf{INSTANCE:} Set $U$ of boolean variables, collection $\mathcal{C}$ of clauses over $U$ such that each clause $c\in \mathcal{C}$ has $|c|=3$ variables and none of the clauses contain complemented variables (\emph{Monotonicity}). 
\item \textbf{QUESTION:} Is there a truth assignment for $U$ such that each clause in $\mathcal{C}$ has at least one true literal and at least one false literal ? 
\end{itemize}	
\end{framed}
It is known that the problem $\textsc{MNAE-3SAT}$ is \textbf{NP}-complete \cite{schaefer1978complexity}. 

\subsubsection{Reduction: \textsc{MNAE-3SAT} $\stackrel{\mathsf{poly}}{\implies} $\textsc{Wireless Broadcast}}
Suppose we are given an instance of the problem \textsc{MNAE-3SAT} $(U,\mathcal{C})$. Let $|U|=n$ and $|\mathcal{C}|=m$. Denote $n$ boolean  variables by $\{x_i, i=1,2, \ldots, n\}$. For this instance of \textsc{MNAE-3SAT}, we consider the following instance $\mathcal{G}(V,E)$ of \textsc{Wireless Broadcast} as shown in Figure \ref{reduction_figure}. The construction is done as follows:

\begin{figure}
\centering
	\begin{overpic}[width=0.45\textwidth]{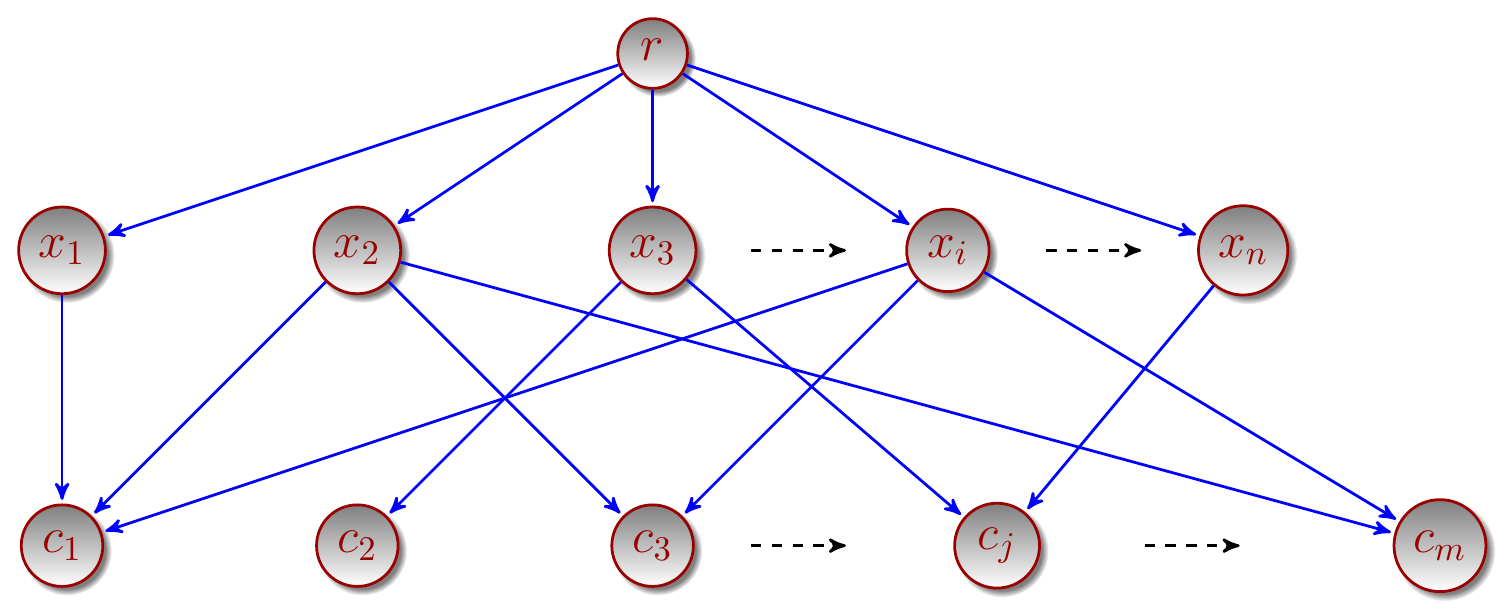}	
	\put(46,39){$\small{C_{\texttt{r}}=2}$}
	%\put(30,34){$2$}
	%\put(56,34){$2$}
	%\put(50,29){$2$}
	%\put(35,29){$2$}
	%\put(44,32){$2$}
	\put(60,28){$C_{x_i}=1$}
	\put(5,28){$C_{x_1}=1$}
	%\put(44,18){$1$}
	%\put(24,18){$1$}
	%\put(64,18){$1$}
	%\put(6,18){$1$}
	%\put(84,18){$1$}
	\end{overpic}
\caption{The Gadget used for the hardness proof}
\label{reduction_figure}
\end{figure}
\begin{itemize}
	\item There are a total of $n+m+1$ nodes. The nodes are divided into three layers as shown in Fig. \ref{reduction_figure}.\\
	\item Let $\texttt{r} \in V$ be the source node in the first layer. The capacity of the source node is $2$.  This means that, the source node can transmit $2$ packets per slot to its out-neighbours. 
	\item There are $n$ nodes in the second layer of the figure \ref{reduction_figure}, all of which are out-neighbors of the source node $\texttt{r}$. Each of these nodes correspond to a variable $x_i$ in \textsc{MNAE-3SAT} instance. Capacity of each of these nodes in the second layer is one. 
	\item There are $m$ nodes in the third layer, each corresponding to a clause $c\in \mathcal{C}$ in the \textsc{MNAE-3SAT} instance. The edges incoming to a node $c_j$ are defined as follows: if the clause $c_j$ is expressed as $c_j=x_{i_1} \vee x_{i_2} \vee x_{i_3} $, then we add three edges $(x_{i_1},c_j), (x_{i_2},c_j), (x_{i_3},c_j)$ in the graph $\mathcal{G}(V,E)$. Capacities of each node in the third layer is taken to be $1$. 
\end{itemize}
Now consider the following instance of \textsf{Wireless Broadcast} on the constructed graph $\mathcal{G}(V,E)$. There are $M=2$ packets at the source with a deadline of $T=2$ slots. 
 We claim that the following two questions are equivalent, i.e. Question $1$ has a YES answer iff the Question $2$ has an YES answer. 
 \begin{itemize}
 	\item \textbf{Question 1:} Is the \textsc{MNAE-3SAT} instance $(U, \mathcal{C})$ satisfiable ?
 \item \textbf{Question 2:} Does the constructed \textsf{Wireless Broadcast} instance have a \textsf{Yes} Solution?
 \end{itemize}
 To show this, let us denote the packets sent by the source $\texttt{r}$ at the beginning of the slot by $\{0,1\}$. Since the capacity of the source $\texttt{r}$ is $2$, all nodes $x_1,x_2, \ldots, x_n$ receive this packet at every slot. Since the capacities of the nodes $x_i$ is only unity, they can only transmit either packet $0$ or the packet $1$ at that slot. We can denote this choice by the binary variable $x_i$, i.e. $x_i=0$ if the node $x_i$ sends packet $0$ and is $1$ if it sends packet $1$.\\
 Note that the node $c_j$ will receive both the packets if the corresponding clause contains at least one $0$ and at least one $1$. For a broadcast capacity of $2$, all nodes must receive both packets at every slot. This is exactly the condition for the existence of a satisfying assignment of the \textsc{MNAE-3SAT} instance. This proves the intended hardness result. $\blacksquare$
\subsubsection*{\textbf{Corollary:}}
As a direct consequence of the above reduction, it follows that the problem \textsf{Wireless Broadcast} remains \textbf{NP}-complete even with the following additional restrictions:
\begin{enumerate}
\item The wireless transmissions are non-interfering.
%, \emph{e.g.}, when each node transmits over a different channel (c.f. \cite{edmonds}). 
\item The graph $\mathcal{G}(V,E)$ is a DAG (c.f. \cite{sinha2015throughput}, \cite{sinha2016throughput2}).
\item The node capacities may take at most two values. 
%\item There are at most two distinct transmission capacities among all nodes.
\item The in-degree of each node is at most $3$.
%\item The number of messages $M$ and the deadline $T$ are bounded by universal constants.
\end{enumerate}

%\newpage

\subsection{Proof of Stability of the Virtual Queues} \label{stability_proof_VQ}
% Consider the Lyapunov function $L(\bm{\tilde{Q}}(t))$, which is defined as the $\ell_2$ norm of the virtual queues, i.e., 
The proof of positive-recurrence and the sample path result is divided into several parts. First we describe and develop some general tools and then apply these tools to the virtual-queue Markov Chain $\{\bm{\tilde{Q}}(t)\}_{t \geq 1}$.  \\
\subsubsection{Mathematical Tools}
The key to our proof is a stronger version of the Foster-Lyapunov drift theorem, obtained by Hajek \cite{hajek1982hitting} in a more general context. The following statement of the result, quoted from \cite{eryilmaz2012asymptotically}, will suffice our purpose. First, we recall the definition of a Lyapunov function:

\begin{definition}[Lyapunov Function]
 Let $\mathcal{X}$ denote the state space of any process. We call a function $L : \mathcal{X} \to \mathbb{R}$ a \textbf{Lyapunov} function if the following conditions hold:\\
 \begin{itemize}
  \item (1)  $ L(x) \geq  0 , \forall x \in \mathcal{X}$ and,
  \item (2) the set $S(M)= \{x \in \mathcal{X} : L(x) \leq M\}$ is finite for all finite $M$.
  \end{itemize}
  \end{definition}
  
%\begin{framed}
\begin{theorem}[Hajek '82] \label{hajek}
 For an irreducible and aperiodic Markov Chain $\{X(t)\}_{t\geq 0}$ over a countable state space $\mathcal{X}$, suppose $L:\mathcal{X} \to \mathbb{R}_+$ is a Lyapunov function. Define the drift of $L$ at $X$ as 
 \begin{equation*}
  \Delta L(X) \stackrel{\Delta}{=} \big(L(X(t+1))-L(X(t))\big) \mathcal{I}(X(t)=X),
 \end{equation*}
where $\mathcal{I}(\cdot)$ is the indicator function. Thus, $\Delta Z(X)$ is a random variable that measures the amount of change in the value of $Z$ in one step, starting from the state $X$. Assume that the drift satisfies the following two conditions:\\
\begin{itemize}
 \item \textbf{(C1)} There exists an $\epsilon >0$ and a $B<\infty$ such that 
 \begin{eqnarray*}
  \mathbb{E}(\Delta L(X)|X(t)=X) \leq -\epsilon, \hspace{5pt}\forall X \in \mathcal{X} \hspace{5pt} \textrm{with} \hspace{3pt} Z(X) \geq B
 \end{eqnarray*}
 \item \textbf{(C2)} There exists a $D < \infty$ such that 
 \begin{eqnarray*}
  |\Delta L(X)| \leq D, \hspace{5pt} \textrm{w.p.}\hspace{2pt} 1, \hspace{5pt} \forall X \in \mathcal{X} 
 \end{eqnarray*}
Then, the Markov Chain $\{X(t)\}_{t \geq 0}$ is positive recurrent. Furthermore, there exists scalars $\theta^* >0$ and a $C^* < \infty$ such that 
\begin{eqnarray*}
 \limsup_{t\to \infty} \mathbb{E}\big(\exp(\theta^* L(X(t))\big) \leq C^* 
\end{eqnarray*}

\end{itemize}

\end{theorem}

We now establish the following technical lemma, which will be useful later.
\begin{framed}
\begin{lemma} \label{tech_lemma}
 Let $\{Y(t)\}_{t \geq 0}$ be a stochastic process taking values on the nonnegative real line. Supppose that there exists scalars $\theta^* >0$ and $C^* < \infty$ such that, 
 \begin{eqnarray} \label{mmt_cnd}
  \limsup_{t \to \infty} \mathbb{E}(\exp(\theta^*Y(t))) \leq C^*
 \end{eqnarray}
Then, 
\begin{eqnarray*}
  Y(t) = \mathcal{O}(\log(t)), \hspace{10pt} \textrm{w.p.1}   
\end{eqnarray*}

\end{lemma}
\end{framed}

\begin{proof}
Define the positive constant $\eta^* = \frac{2}{\theta^*}$. We will show that 
\begin{eqnarray*}
 \mathbb{P}( Y(t) \geq \eta^* \log(t), \hspace{10pt} \textrm{i.o.}\footnote{\emph{i.o}=\emph{infinitely often}})=0.  
\end{eqnarray*}
For this, define the event $E_t$ as 
\begin{eqnarray} \label{over_shoot_event}
 E_t = \{ Y(t) \geq \eta^* \log(t) \}
\end{eqnarray}
From the given condition \eqref{mmt_cnd}, we know that there exists a finite time $t^*$ such that 
\begin{eqnarray} \label{bdd_moment}
 \mathbb{E} (\exp(\theta^*Y(t))) \leq C^*+1, \hspace{5pt} \forall t \geq t^* 
\end{eqnarray}
We can now upper-bound the probabilities of the events $E_t, t \geq t^*$ as follows
\begin{eqnarray*}
 \mathbb{P}(E_t)&=&\mathbb{P}( Y(t) \geq \eta^* \log(t) ) \\
 &=& \mathbb{P}\big( \exp(\theta^*Y(t)) \geq \exp(\theta^* \eta^* \log(t)) \big)\\
 &\stackrel{(a)}{\leq}& \frac{\mathbb{E}(\exp(\theta^*Y(t)))}{t^2} \\
 &\stackrel{(b)}{\leq}& \frac{C^*+1}{t^2}
\end{eqnarray*}
The inequality (a) follows from the Markov inequality and the fact that $\theta^*\eta^*=2$. The inequality (b) follows from Eqn. \eqref{bdd_moment}. Thus, we have 
\begin{eqnarray*}
 \sum_{t=1}^{\infty}\mathbb{P}(E_t) &=& \sum_{t=1}^{t^*-1} \mathbb{P}(E_t) + \sum_{t=t^*}^{\infty}\mathbb{P}(E_t)\\
 &\leq & t^* + (C^*+1)\sum_{t=t^*}^{\infty} \frac{1}{t^2}\\
 & \leq& t^* + (C^*+1) \frac{\pi^2}{6} < \infty 
\end{eqnarray*}
Finally, using the Borel Cantelli Lemma, we conclude that 
\begin{eqnarray*}
 \mathbb{P}(\limsup Y_t \geq \eta^* \log t) = \mathbb{P}(E_t \hspace{5pt} \textrm{i.o.}) = 0 
\end{eqnarray*}
This proves that $ Y_t= \mathcal{O}(\log t), \textrm{w.p.} 1$. 
\end{proof}
Combining Theorem \ref{hajek} with Lemma \ref{tech_lemma}, we have the following corollary 
\begin{framed}
\begin{corollary} \label{corollary}
 Under the conditions (C1) and (C2) of Theorem \ref{hajek}, we have 
 \begin{eqnarray*}
  L(X(t)) = \mathcal{O}(\log t) , \hspace{10pt} \textrm{w.p.} \hspace{3pt} 1
 \end{eqnarray*}

\end{corollary}
\end{framed}
\subsubsection{Construction of a Stationary Randomized Policy for the Virtual Queues $\{\bm{\tilde{Q}}(t)\}_{t\geq 1}$} \label{stationary_policy}
Let $\mathcal{D}$ denote the set of all Connected Dominating Sets (CDS) in the graph $\mathcal{G}$ containing the source $\texttt{r}$. Since the broadcast rate $\lambda < \lambda^*$ is achievable by a stationary randomized policy, there exists such a policy $\pi^{*}$ which  executes the following:
\begin{itemize}
\item  There exist non-negative scalars $\{a_i^*, i =1,2,\ldots, |\mathcal{D}|\}$ with $\sum_i a_i^*=\lambda$, such that each new incoming packet is routed independently along a CDS $D_i \in \mathcal{D}$ with probability $\frac{a_i^*}{\lambda}, \forall i$. The packet routed along the CDS $D_i$ corresponds to an arrival to the virtual queues $\{Q_j, j \in D_i\}$. \\
As a result, packets arrive to the virtual queue $Q_j$ i.i.d. at an expected rate of $\sum_{i:j \in D_i} a^*_j, \forall j$ per slot.
%Thus, the expected number of arrivals contributed by the CDS $D_i$ to the virtual queues corresponding to its vertices is $a_i$ packets per slot. 
\item  A feasible schedule $\bm{s}_j \in \mathcal{M}$ is selected for transmission with probability $p_j$ $j=1,2,\ldots, k$ i.i.d. at every slot. By Caratheodory's theorem, the value of $k$ can be restricted to at most $n+1$. This results in the following expected service rate vector from the virtual queues:  
\begin{eqnarray*}
	\bm{\mu}^*= \sum_{j=1}^{n+1}p_j c_j\bm{s}_j,
\end{eqnarray*}
%where $\bm{p}$ is a probability distribution vector and $\bm{s}_i \in \mathcal{M}_i, \forall i=1,2, \ldots, n+1$.\\
%With this decomposition, The activation $\bm{s}_i$  chosen with probability $p_i, \forall i$. Hence the expected service rate of the $i$\textsuperscript{th} server of the virtual queue $Q_i$ is $\bm{\mu}^*_i$ at all slots. 
\end{itemize}
Since each of the virtual queues must be stable under the action of the policy $\pi^{*}$, from the theory of the GI/GI/1 queues, we know that there exists an $\epsilon >0$ such that 
\begin{eqnarray} \label{stability_cond}
 \mu_i^*-\sum_{j: i\in D_j}a_j^* \geq \epsilon, \hspace{10pt} \forall i \in V
\end{eqnarray}

Next, we will verify that the conditions \textbf{C1} and \textbf{C2} in Theorem \ref{hajek} holds for the Markov Chain of the virtual queues $\{\bm{\tilde{Q}}(t)\}_{t\geq 1}$ under the action of the \textbf{UMW} policy, with the Lyapunov function $L(\bm{\tilde{Q}}(t))= ||\bm{\tilde{Q}}(t)||$ at any arrival rate $\lambda < \lambda^*$. 
\subsubsection{Verification of Condition (\textbf{C1})- Negative Expected Drift}
From the definition of the policy \textbf{UMW}, it minimizes the RHS of the drift upper-bound \eqref{UMW_def} from the set of all feasible policies $\Pi$. Hence, we can upper-bound the conditional drift of the \textbf{UMW} policy by comparing it with the stationary policy $\pi^*$ described in \ref{stationary_policy} as follows:  
\begin{eqnarray*}
	&&\mathbb{E}(\Delta^{\textbf{UMW}}(t)|\bm{\tilde{Q}}(t)=\bm{\tilde{Q}})\nonumber \\
	&{\leq}& \frac{1}{2||\bm{\tilde{Q}}||}\bigg( B+  2 \sum_{i\in V}\tilde{Q}_i(t)\mathbb{E}\big(A_i^{\textrm{UMW}}(t)|\bm{\tilde{Q}}(t)=\tilde{\bm{Q}}\big)\nonumber
	\end{eqnarray*}
	\begin{eqnarray*}
	&-& 2 \sum_{i\in V}\tilde{Q}_i(t)\mathbb{E}\big(\mu_i^{\textrm{UMW}}(t)|\bm{\tilde{Q}}(t)=\tilde{\bm{Q}}\big)\bigg)
	\end{eqnarray*}
	\begin{eqnarray*}
	&\stackrel{(a)}{\leq}& \frac{1}{2||\bm{\tilde{Q}}||}\bigg( B+  2 \sum_{i\in V}\tilde{Q}_i(t)\mathbb{E}\big(A_i^{\pi^*}(t)|\bm{\tilde{Q}}(t)=\tilde{\bm{Q}}\big)\nonumber\\
	&-& 2 \sum_{i\in V}\tilde{Q}_i(t)\mathbb{E}\big(\mu_i^{\pi^*}(t)|\bm{\tilde{Q}}(t)=\tilde{\bm{Q}}\big)\bigg)\nonumber\\
	\end{eqnarray*}
	\begin{eqnarray*}
	&=& \frac{1}{2||\bm{\tilde{Q}}||}\bigg(B - 2 \sum_{i\in V} \tilde{Q}_i(t) \big(\mathbb{E}\mu_i^*(t)-\mathbb{E}A_i^*(t)\big)\bigg) \nonumber
	\end{eqnarray*}
	\begin{eqnarray*}
	&=& \frac{1}{2||\bm{\tilde{Q}}||} \bigg(B - 2 \sum_{i\in V} \tilde{Q}_i(t) \big(\mu_i^*-\sum_{j: i\in D_j}a_j^* \big)\bigg)\nonumber\\
	\end{eqnarray*}
	\begin{eqnarray}
	&\stackrel{(b)}{\leq }& \frac{B}{2||\bm{\tilde{Q}}||}-\epsilon \frac{\sum_{i \in V} \tilde{Q}_i(t)}{||\bm{\tilde{Q}}||}, \label{cur_bd}
\end{eqnarray}
where inequality (a) follows from the definition of the \textbf{UMW} policy and inequality (b) follows from the stability property of the randomized policy given in Eqn. \eqref{stability_cond}. 
Since the virtual-queue lengths $\bm{\tilde{Q}}(t)$ is a non-negative vector, it is easy to see that (\emph{e.g.} by squaring both sides) 
\begin{eqnarray*}
	\sum_{i \in V} \tilde{Q}_i(t) \geq \sqrt{\sum_i \tilde{Q}^2_i(t)}= ||\bm{\tilde{Q}}||
\end{eqnarray*}
Hence, from Eqn. \eqref{cur_bd} in the above chain of inequalities, we obtain
\begin{eqnarray} \label{drift_ineq3}
	\mathbb{E}(\Delta^{\textbf{UMW}}(t)|\bm{\tilde{Q}}(t)=\bm{\tilde{Q}}) \leq  \frac{B}{2||\bm{\tilde{Q}}||}-\epsilon
\end{eqnarray}
Thus, 
\begin{eqnarray*}
	\mathbb{E}(\Delta^{\textbf{UMW}}(t)|\bm{\tilde{Q}}(t)) \leq -\frac{\epsilon}{2}, \hspace{20pt} \forall ||\bm{\tilde{Q}}|| \geq B/\epsilon
\end{eqnarray*}
This verfies the negative expected drift condition \textbf{C1} in Theorem \ref{hajek}. 

\subsubsection{Verification of Condition (\textbf{C2})- Almost Surely Bounded Drift }
To show that the magnitude of one-step drift $|\Delta L(\bm{\tilde{Q}})|$ is almost surely bounded, we compute 
\begin{eqnarray*}
 |\Delta L(\bm{\tilde{Q}}(t))| &=& |L(\bm{\tilde{Q}}(t+1))- L(\bm{\tilde{Q}}(t))|\\
 &=& \big|||\bm{\tilde{Q}}(t+1)||- ||\bm{\tilde{Q}}(t))||\big|
\end{eqnarray*}
Now, from the dynamics of the virtual queues \eqref{queue-evolution}, we have for any virtual queue $\tilde{Q}_i$:
\begin{eqnarray*}
  |\tilde{Q}_i(t+1)- \tilde{Q}_i(t)| \leq |A_i(t)-\mu_i(t)|
\end{eqnarray*}
Thus,
\begin{eqnarray*}
 ||\bm{\tilde{Q}}(t+1)- \bm{\tilde{Q}}(t)|| \leq ||\bm{A}(t) - \bm{\mu}(t)|| \leq \sqrt{n}(A_{\max} + c_{\max})
\end{eqnarray*}
Hence, using the triangle inequality for the $\ell_2$ norm, we obtain
\begin{eqnarray*}
 |\Delta L(\bm{\tilde{Q}}(t))|=\big|||\bm{\tilde{Q}}(t+1)||- ||\bm{\tilde{Q}}(t))||\big|\leq \sqrt{n}(A_{\max} + c_{\max}),
\end{eqnarray*}
which verifies the condition \textbf{C2} of Theorem \ref{hajek}. 
\subsubsection{Almost Sure Bound on Virtual Queue Lengths}
Finally, we invoke Corollary \ref{corollary} to conclude that 
\begin{eqnarray*}
 \limsup_t ||\bm{\tilde{Q}}(t)|| = \mathcal{O}(\log t), \hspace{10pt} \textrm{w.p.} 1
\end{eqnarray*}
This implies that, 
\begin{eqnarray*}
 \max_i \bm{\tilde{Q}}_i(t) = \mathcal{O}(\log t) , \hspace{10pt} \textrm{w.p.} 1 \hspace{10pt} \blacksquare
\end{eqnarray*}

\subsection{Proof of Theorem \ref{physical_stability-theorem} } \label{physical_stability-proof}
Throughout this proof, we will consider only the typical sample point $\omega \in \mathcal{E}$ defined in Eqn. \eqref{q_bd}. For the sake of notational simplicity, we will drop the argument $\omega$ for evaluating a random variable $\bm{X}$ at $\omega$, i.e., the deterministic sample path $X(\omega, t), \omega \in \mathcal{E}$ will be simply denoted by $X(t)$. We now make a simple observation which will be useful in the proof of the theorem:
\begin{framed}
\begin{lemma} \label{seq_lemma}
 Consider a function $F: \mathbb{Z}_+ \to \mathbb{Z}_+$, where $\mathbb{Z}_+$ is the set of non-negative integers. Assume that $F(t)=\mathcal{O}(\log t)$. Define $M(t)= \sup_{0\leq \tau \leq t} F(\tau)$. Then,
 \begin{enumerate}
 \item $M(t)$ is non-decreasing in $t$ and $M(t) \geq F(t)$.
 \item $M(t)=\mathcal{O}(\log t)$.
 \end{enumerate}
\end{lemma}
\end{framed}
\begin{proof}
Clearly, $M(t)\sup_{0\leq \tau \leq t} F(\tau) \geq F(t)$ and 
\begin{eqnarray*}
M(t+1)=\sup_{0\leq \tau \leq t+1}F(\tau) \geq \sup_{0\leq \tau \leq t}F(\tau)= M(t). 
\end{eqnarray*}
To prove the second claim, let $t_{\max}(t)= \arg \max_{0\leq \tau \leq t} F(\tau)$. Clearly, $t_{\max}(t) \leq t$. Hence, for large enough $t$, we have
\begin{eqnarray*}
M(t) = F(t_{\max}(t))= \mathcal{O}(\log t_{\max}(t)) = \mathcal{O}(\log t).
\end{eqnarray*}
\end{proof}

As a consequence of Lemma \ref{seq_lemma} applied to Eqn. \eqref{arr_bd}, we have almost surely
\begin{framed}
\begin{eqnarray} \label{arr_cond}
 A_i(t_0,t) \leq S_i(t_0,t) + M(t), \hspace{10pt} \forall i \in V, \hspace{3pt}\forall t_0 \leq t,
\end{eqnarray}
\end{framed}
for some non-decreasing function $M(t)=\mathcal{O}(\log t)$. 
We now return to the proof of the main the main result, Theorem \ref{physical_stability-theorem}. 
% \begin{framed}
% \begin{proposition}
%  \textbf{{UMW}} is rate stable.
%  \end{proposition}
% \end{framed}
\begin{proof}[of the main result]
Our proof technique is inspired by an adversarial queueing theory argument, given in \cite{Gamarnik}. We remind the reader that we are analyzing a typical sample path satisfying Eqn. \eqref{arr_cond}, which holds almost surely. In the following argument, each copy of a packet is counted separately.\\
Without any loss of generality, assume that we start from an empty network at time $t=0$. Let $R_k(t)$ denote the total number of packets waiting to be transmitted further at time $t$, which have already been forwarded \emph{exactly} $k$ times by the time $t$. We call such packets ``layer $k$" packets. As we have mentioned earlier, if a packet is duplicated multiple times along its assigned route $D$ (which is a connected dominating set (or CDS, in short)), each copy of the packet is counted separately in the variable $R_k(t)$, \emph{i.e.,}
\begin{eqnarray} \label{r-k}
R_k(t)=\sum_{D \in \mathcal{D}} \sum_{i \in D_k} R_{(i, D)}(t), 
\end{eqnarray} 
where the variable $R_{(i,D)}(t)$ denotes the number of packets following the CDS $D$, that are waiting to be transmitted by the node $i \in D$ at time $t$ and $D_k$ is the set of nodes in the CDS $D$, which are exactly $k$\textsuperscript{th} hop away from the source along the CDS $D$. 
%If there are more than one node, we include all these nodes in the summation \eqref{r-k}.
We show by induction that $R_k(t)$ is \emph{almost surely} bounded by a function, which is $\mathcal{O}(\log t)$.\\\\
\indent \textbf{Base Step} $k=0$: Consider the source node $i= \texttt{r}$ and an arbitrary time $t$. Let $t_0\leq t$ be the largest time at which no packets of layer $0$ (packets which are present only at the source and have never been transmitted before) were waiting to be transmitted by the source. If no such time exists, set $t_0=0$. 
%Hence, the total number of layer $0$ packets waiting to be transmitted by the source at time $t_0$ is at most $Q_{\texttt{r}}(0)$, which can be assumed to be zero, without any loss of generality. 
During the time interval $(t_0,t]$, as a consequence of the property in Eqn.\ \eqref{arr_cond} of the \textbf{UMW} policy, at most $S_{\texttt{r}}(t_0,t)+M(t)$ external packets have arrived to the source $\texttt{r}$ for broadcasting. Also, by the choice of the time $t_0$, the source node $\texttt{r}$ was always having packets to transmit during the entire time interval $(t_0,t]$. Since \textsf{LTF} packet scheduling policy is followed in the physical network, layer $0$ packets have priority over all other packets (in fact, there is packet of other layers present \emph{only} at the source, but this is not the case at other nodes which we will consider in the induction step).  Hence, it follows that the total number of layer $0$ packets at time $t$ satisfies 
\begin{eqnarray}
R_0(t)= \sum_{D \in \mathcal{D}} \sum_{i \in {D}_0} R_{(i,D)} (t) &\leq&  S_{\texttt{r}}(t_0,t)+M(t) - S_{\texttt{r}}(t_0,t) \nonumber \\
&\leq&  M(t) 
\end{eqnarray}
Define $B_0(t)\stackrel{\text{def}}{=}M(t)$. Since $M(t)=\mathcal{O}(\log t)$, we have $B_0(t)=\mathcal{O}(\log t)$. Note that, since $M(t)$ is non-decreasing from Lemma \eqref{seq_lemma}, so is $B_0(t)$. \\\\
\textbf{Induction Step:}
As our inductive assumption, suppose that, for some non-decreasing functions $B_j(t)=\mathcal{O}(\log t), j=0,1,2,\ldots, k-1$, we have $ R_j(t) \leq B_j(t)$, for all time $t$. We next show that there exists a non-decreasing function $B_k(t)=\mathcal{O}(\log t)$ such that $ R_k(t) \leq B_k(t)$ for all time $t$. \\
To prove the above assertion, fix a node $i$ and an arbitrary time $t$. Let $t_0\leq t$ denote the largest time before $t$, such that there were no layer $k$ packets waiting to be transmitted by the node $i$. Set $t_0=0$ if no such time exists. Hence the node $i$ was always having packets to transmit during the time interval $(t_0,t]$ (packets in layer $k$ or lower). The layer $k$ packets that wait to be transmitted by the node $i$ at time $t$ are composed only of a subset of packets which were in layers $0\leq j \leq k-1$ at time $t_0$ or packets that arrived during the time interval $(t_0, t]$ and include the node $i$ as \emph{one of their $k$\textsuperscript{th} transmitter} along the route followed. By our induction assumption, the first group of packets has a size bounded by $\sum_{j=0}^{k-1} B_j(t_0)\leq \sum_{j=0}^{k-1}B_j(t)$, where we have used the fact (using our induction step) that the functions $B_j(\cdot)$'s are monotonically non-decreasing. The size of the second group of packets is given by $\sum_{D : i \in D_k} A_D(t_0,t)$. We next estimate the number of layer $k$ packets that crossed the edge $e$ during the time interval $(t_0,t]$. Since the \textsf{LTF} packet scheduling policy is used in the physical network, layer $k$ packets were not processed only when there were packets in layers up to $k-1$ that included the node $i$ in its routing CDS. The number of such packets is bounded by $\sum_{j=0}^{k-1}B_j(t_0)\leq \sum_{j=0}^{k-1}B_j(t)$, which denotes the total possible number of packets in layers up to $k-1$ at time $t_0$, \emph{plus} $\sum_{j=0}^{k-1} \sum_{D: i \in D_j}A_D(t_0,t)$, which is the number of new packets that arrived in the interval $(t_0,t]$ and includes the node $i$ as a transmitter within their first $k-1$ hops. Thus, we conclude that at least 
\begin{eqnarray}
\max \bigg\{0, S_i(t_0,t) - \sum_{j=0}^{k-1}B_j(t) - \sum_{j=0}^{k-1} \sum_{D: i \in D_j}A_D(t_0,t) \bigg \}
\end{eqnarray}  
packets of layer $k$ have been transmitted by the node $i$ during the time interval $(t_0,t]$. Hence, the total number of layer $k$ packets present at node $i$ at time $t$ is given as 
\begin{eqnarray*}
&\sum_{D: i \in D_k}\hspace{-10pt}& R_{(i,D)}(t) \leq \sum_{j=0}^{k-1} B_j(t) + \sum_{D : i \in D_k} A_D(t_0,t) \\
&-& \big(S_i(t_0,t) - \sum_{j=0}^{k-1}B_j(t) - \sum_{j=0}^{k-1} \sum_{D: i \in D_j}A_D(t_0,t)\big)\\
&=& 2\sum_{j=0}^{k-1}B_j(t) + \sum_{j=0}^{k} \sum_{D: i \in D_j}A_D(t_0,t) - S_i(t_0,t) \\
&\stackrel{(a)}{\leq} & 2\sum_{j=0}^{k-1}B_j(t) + A_i(t_0,t) - S_i(t_0,t) \\
& \stackrel{(b)}{\leq} & 2\sum_{j=0}^{k-1}B_j(t) + M(t),
\end{eqnarray*}
where the inequality (a) follows from the fact that each packet gets routed to a node $i$ for transmission only once and hence 
\begin{eqnarray*}
	A_i(t_0,t)= \sum_{j=0}^{n-1}\sum_{D: i \in D_j}A_D(t_0,t), \hspace{10pt} \forall i\in V. 
\end{eqnarray*}
The inequality (b) follows from the property of the typical sample paths, stated in Eqn.\ \eqref{arr_cond}. Hence, the total number of layer $k$ packets at time $t$ is bounded as
\begin{eqnarray*}
R_k(t) = \sum_i \sum_{D: i \in D_k}\hspace{-10pt}& R_{(i,D)}(t) \leq 2n\sum_{j=0}^{k-1}B_j(t) + nM(t)
\end{eqnarray*}
Define $B_k(t)$ to be the RHS of the above equation, i.e. 
\begin{eqnarray}\label{bk}
B_k(t) \stackrel{(\text{def})}{=}  2n\sum_{j=0}^{k-1}B_j(t) + nM(t)
\end{eqnarray} 
 
Using our induction assumption and Eqn. \eqref{bk}, we conclude that 
\newpage 
$B_k(t)=\mathcal{O}(\log t)$, and it is easily seen to be non-decreasing. This completes the proof of the induction step. 

To conclude the proof of the theorem, observe that the sum of the lengths of the physical queues at time $t$ may be alternatively written as 
\begin{eqnarray}\label{sumQ}
\sum_{i \in V}{Q}_i(t)=\sum_{k=1}^{n-1} R_k(t)
\end{eqnarray}
Since the previous inductive argument shows that for all $k$, we have $R_k(t) \leq B_k(t)$ where $B_k(t)=\mathcal{O}(\log t)$ \emph{a.s.}, we have $\sum_{i \in V}{Q}_i(t)=\mathcal{O}(\log t)$, and hence
%Thus combining Eqn. \eqref{sumQ} with Lemma \eqref{rate_stability}, we conclude that  
\begin{eqnarray}
\lim_{t \to \infty} \frac{\sum_{i \in V}{Q}_i(t)}{t}=0, \hspace{15pt} \text{w.p. } 1.
\end{eqnarray}
%This implies that the physical queues are rate stable \cite{neely2010stochastic}, jointly under the operation of \textbf{UMW} and \textsf{LTF}. 
\end{proof}

\subsection{Proof of Lemma \ref{cap_bd1}} \label{cap_bd_proof}
\begin{proof}
Observe that, due to the primary interference constraints, the nodes $\texttt{r}, a$ and $d$ can not be activated at the same slot. Consider any arbitrary policy, which activates the node $i$ for a fraction $f_i, i \in V$ times. Hence, we have the constraint that 
\begin{eqnarray} \label{constr1}
 f_{\texttt{r}}+f_a+f_c\leq 1
\end{eqnarray}
On the other hand, if the policy $\pi$ achieves a broadcast rate of $\lambda$, it must be that 
\begin{eqnarray*}
 \lambda &\leq& f_\texttt{r}, \hspace{10pt} \textrm{ (considering broadcast rate at node } a)\\
 \lambda &\leq& f_a, \hspace{10pt} \textrm{ (considering broadcast rate at node } b) \\
 \lambda &\leq& f_c, \hspace{10pt} \textrm{ (considering broadcast rate at node } f).
\end{eqnarray*}
Adding the above three equations, we have 
\begin{eqnarray}
 3\lambda \leq f_{\texttt{r}}+f_a+f_c \stackrel{(a)}{\leq} 1, 
\end{eqnarray}
where the inequality (a) follows from the constraint \eqref{constr1}. Since the policy $\pi$ is assumed to be arbitrary, we conclude that the broadcast capacity of the $3\times 3$ grid network is at most $\frac{1}{3}$. 
\end{proof}

\end{document}